\documentclass[11pt]{article}
\usepackage[a4paper, margin=1in]{geometry}
\newif\ifdraft%
\usepackage{graphicx}
\usepackage{xcolor,xspace}
\usepackage{amsmath}
\usepackage{amssymb}
\usepackage{amsthm}
\usepackage{caption}
\usepackage{pgfplots}
\usepackage{subcaption}
\usepackage{tikz}
\usepackage{pgfmath} 
\usetikzlibrary{calc,angles,positioning,intersections}
\usetikzlibrary{patterns}
\usetikzlibrary {arrows.meta}
\usetikzlibrary{fit}
\definecolor{ocre}{RGB}{255,0,0}
\usepackage[nopatch=footnote]{microtype}
\usepackage{mathtools}
\usepackage{paralist}
\usepackage{thmtools,thm-restate}
\usepackage{scrextend}
\usepackage{tcolorbox}
\pgfplotsset{compat=1.7}
\usepackage{algorithm}
\usepackage{algcompatible}
\usepackage{enumitem} 
\usepackage[unicode,hidelinks, hypertexnames=false]{hyperref}
\hypersetup{
    colorlinks=true,
    linkcolor=blue,
    citecolor=blue,
    filecolor=black,
    urlcolor=[HTML]{0088CC},
    pdftitle={Towards Optimal Distributed Delta Coloring},
    pdfauthor={Anonymous}
}

\usepackage[capitalize, nameinlink]{cleveref}
\crefname{ALG@line}{Step}{Steps}
\Crefname{step}{Step}{Steps}
\newcommand{\myStepCref}[1]{\hyperref[#1]{Step~\ref{#1}}}

\newcommand{\Vsparse}{\ensuremath{V_{sparse}}\xspace}
\usepackage{changepage}

\newcommand{\cL}{\mathcal{L}}

\newcommand{\cS}{\mathcal{S}}

\newcommand{\CONGEST}{\ensuremath{\mathsf{CONGEST}}\xspace}
\newcommand{\LOCAL}{\ensuremath{\mathsf{LOCAL}}\xspace}
\newcommand{\lovasz}{Lov\'{a}sz\xspace}

\newcommand{\hhc}{\ensuremath{Q}}

\usepackage{todonotes}
\theoremstyle{plain}
\newtheorem{lemma}{Lemma}
\newtheorem{observation}[lemma]{Observation}
\newtheorem{claim}[lemma]{Claim}
\newtheorem{definition}[lemma]{Definition}
\newtheorem{corollary}[lemma]{Corollary}

\newcommand{\eps}{\varepsilon}
\newcommand{\poly}{\operatorname{\text{{\rm poly}}}}

\newcommand{\vhard}{\ensuremath{\mathcal{V}_{hard}}}
\newcommand{\ehard}{\ensuremath{\mathcal{E}_{hard}}}
\newcommand{\hardcliques}{\ensuremath{\mathcal{C}_{hard}}\xspace}
\newcommand{\hardcliqueshso}{\ensuremath{\mathcal{C}_{HEG}}\xspace}
\newcommand{\typeI}{\ensuremath{Type~I}\xspace}
\newcommand{\typeII}{\ensuremath{Type~II}\xspace}
\newcommand{\typeIplus}{\ensuremath{Type~I^+}\xspace}

\newcommand{\epsvalue}{1/63}
\newcommand{\deltavalue}{28}
\usepackage{adjustbox}

\newcommand*\samethanks[1][\value{footnote}]{\footnotemark[#1]}

\title{Towards Optimal Distributed Delta Coloring}

\author{
	Manuel Jakob \thanks{
		This research was funded in whole or in part by the Austrian Science Fund (FWF) \url{https://doi.org/10.55776/P36280}. For open access purposes, the author has applied a CC BY public copyright license to any author-accepted manuscript version arising from this submission.
	}
	\and Yannic Maus \samethanks}
\date{}
\begin{document}

\maketitle

\thispagestyle{empty}

\section*{Abstract}
The $\Delta$-vertex coloring problem has become one of the prototypical problems for understanding the complexity of local distributed graph problems on constant-degree graphs. The major open problem is whether the problem can be solved deterministically in logarithmic time, which would match the lower bound [Chang et al., FOCS'16].  Despite recent progress in the design of efficient $\Delta$-coloring algorithms, there is currently a polynomial gap between the upper and lower bounds.

In this work we present a $O(\log n)$-round deterministic $\Delta$-coloring algorithm for dense constant-degree graphs, matching the lower bound for the problem on general graphs. For general $\Delta$ the algorithms' complexity is $\min\{\widetilde{O}(\log^{5/3}n),O(\Delta+\log n)\}$. All recent distributed and sublinear graph coloring algorithms (also for coloring with more than $\Delta$ colors)  decompose the graph into sparse and dense parts. Our algorithm works for the case that this decomposition has no sparse vertices. Ironically, in recent (randomized) $\Delta$-coloring algorithms, dealing with sparse parts was relatively easy and these dense parts arguably posed the major hurdle. We present a solution that addresses the dense parts and may have the potential for extension to sparse parts.

Our approach is fundamentally different from prior deterministic algorithms and hence hopefully contributes towards designing an optimal algorithm for the general case. Additionally, we leverage our result to also obtain a randomized $\min\{\widetilde{O}(\log^{5/3}\log n), O(\Delta+\log\log n)\}$-round algorithm for $\Delta$-coloring dense graphs that also matches the lower bound for the problem on general constant-degree graphs [Brandt et al.; STOC'16].

\clearpage
\thispagestyle{empty}
\tableofcontents
\clearpage

\setcounter{page}{1}
\section{Introduction}
In the 1940ies Brooks showed that any connected graph with maximum degree $\Delta$ admits a proper vertex coloring with $\Delta$ colors, unless it is a clique on $\Delta+1$ vertices or an odd cycle \cite{brooks_1941}. Recently, this problem has received ample attention in various sublinear models of computation, e.g., in streaming \cite{AKM22}, massively parallel computing \cite{CCDM24}, or most prominently in distributed message passing models \cite{PS95,LLL_lowerbound,Chang2016a,GHKM18,GHKM21,BBKO2021hideandseek,MU21,HM24}. We continue with introducing the standard distributed message passing model and explain thereafter why the $\Delta$-coloring problem is of particular interest to the field.

\vspace{-0.3cm}
\paragraph{The \LOCAL model of distributed computing \cite{linial92}.} A communication network is modeled as an $n$-node graph, where vertices represent computing entities and edges serve as communication channels. Communication occurs in synchronous rounds, with each node sending unbounded messages to neighbors. Nodes have unlimited computation time, and an algorithm's round complexity is the number of rounds until every node outputs its part in a globally consistent solution, e.g., its color.
\smallskip

Efficient distributed algorithms must progress in large parts of the input graph in parallel. Prime examples where this has been achieved are greedy problems like maximal independent set, maximal matching, and $\Delta+1$ coloring, all solvable sequentially with a trivial greedy approach. For instance, in $\Delta+1$ coloring, each node always has an available color, regardless of prior assignments. Greedy problems are well-suited for parallelization, yielding highly efficient algorithms. On constant-degree graphs, these problems can be solved in $\Theta(\log^* n)$ rounds \cite{linial92,panconesi-rizzi,Naor91}.

The $\Delta$-coloring problem is fundamentally different. While coloring with a single color less is probably of no importance to any application, the problem is of an entirely different nature. If we process many nodes in parallel and greedily color large parts of the graph, we are likely to run into situations where extending the coloring to a valid $\Delta$-coloring is impossible.

Thus, even on constant-degree graphs, the $\Delta$-coloring problem has a $\Omega(\log n)$ lower bound for  deterministic round complexity  and an $\Omega(\log\log n)$ lower bound for  randomized round complexity \cite{Chang2016a,LLL_lowerbound}. Understanding, whether these lower bounds are tight (at least on constant-degree graphs) has become one of the major open problem of the field, see, e.g., \cite{BBKO2021hideandseek}:
\begin{center}
	\emph{Is there a logarithmic-time deterministic distributed algorithm for the $\Delta$-coloring problem?}
\end{center}
In this paper, we make partial progress towards answering this question affirmatively. Recent graph coloring algorithms in distributed and other models rely on decomposing the graph into sparse and dense vertices, which are largely treated separately. A vertex is sparse if it has many non-edges in its neighborhood, and dense otherwise. A graph is dense if it contains no sparse vertices. We discuss this sparse-dense decomposition at the end of the section.

\begin{tcolorbox}
	\textbf{Main Contribution I:} There is a deterministic logarithmic-time distributed algorithm for $\Delta$-coloring constant-degree dense graphs (see \Cref{thm:DeltaColoring}).
\end{tcolorbox}
The actual runtime of our algorithm is $O\min\{\widetilde{O}(\log^{5/3}n),O(\Delta+\log n)\}$.\footnote{$\widetilde{O}(f(n))$ hides factors that are poly-logarithmic in $f(n)$.} On constant-degree graphs the complexity of our algorithm matches the lower bound for the problem on general graphs \cite{Chang2016a}. It is intriguing to note that the lower bound holds even for the sparsest graphs, namely trees, while our upper bound holds for dense graphs. Proving lower bounds for dense graphs remains a significant challenge, as we currently lack effective tools. For example, the influential round elimination lower bound technique applies only to trees and does currently not extend to graphs with cycles \cite{Brandt19}. Conceptually, it may be that dense graphs could be colored in sublogarithmic time, but as we only see our algorithm (and its conceptual ideas) as one part of a solution to the general problem, allowing logarithmic time to color dense graphs is reasonable.
\begin{figure}
	\includegraphics[width=\textwidth]{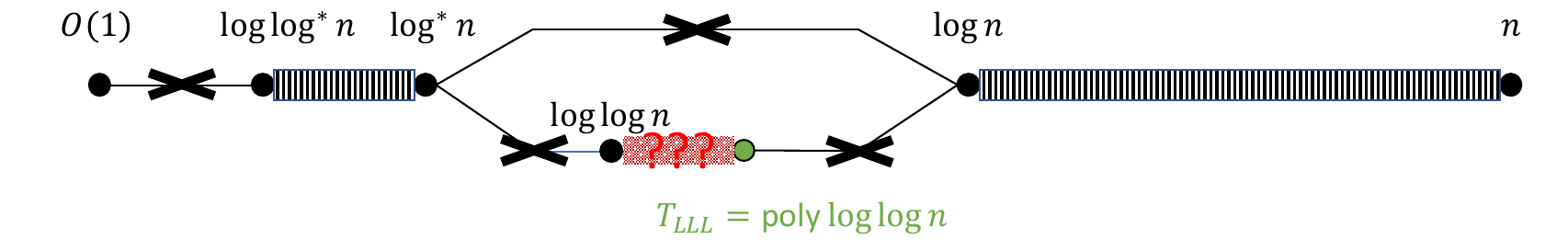}
	\caption{The complexity landscape of LCL problems: Each black dot represents a problem with the depicted asymptotic complexity. For $T_{LLL}$, the randomized lower bound is $\Omega(\log\log n)$, and the best upper bound is $O(\poly\log\log n)$. The top branch reflects deterministic complexities, the bottom branch randomized ones. An $\mathsf{X}$ marks impossible complexities, e.g., no problem has deterministic complexity $\omega(\log^*n)$ and $o(\log n)$. Infinitely many complexities exist in the shaded areas. Current research focuses on the poorly understood red-shaded area, where $\Delta$-coloring is the only natural problem with its best upper bound inside. }		\label{fig:LCLLandscape}
	\vspace{-0.3cm}
\end{figure}
\vspace{-0.3cm}
\paragraph{Why should one care about $\Delta$-coloring and its randomized complexity? } A long line of research mapped out the complexity landscape of so-called locally checkable labeling problems (LCLs), e.g, \cite{naor95,BBHORS21,BBOS18,LLL_lowerbound,BHKLOS18lclComplexity,binary_lcls,Brandt19,Changhierarchy19,Chang2016a,Olivetti2019REtor,BCMOS21,MU21}. In a nutshell LCLs are problems on constant-degree graphs for which a $O(1)$-round \LOCAL algorithm can verify a solution. Examples are many graph coloring problems or the maximal independent set problem.  See \Cref{fig:LCLLandscape} for an illustration of the current state of the art.  There is one central part in the complexity landscape that is extremely poorly understood. These are the problems that admit sublogarithmic-time randomized algorithms, but no sublogarithmic-time deterministic algorithms.  In a seminal result Chang and Pettie showed that any such problem can be solved in the same time as the \lovasz Local Lemma (LLL) \cite{Changhierarchy19}. They also conjectured that LLL on constant-degree graphs can be solved in $O(\log\log n)$ rounds, matching its lower bound established nearly a decade ago \cite{LLL_lowerbound}. Currently, upper bounds are still polynomially larger and used techniques can provably not close the gap \cite{FGLLL17,GG24}.

The  state-of-the-art randomized complexity of the $\Delta$-coloring problem on constant-degree graphs is $O(\log^2\log n)$, making it, to the best of our knowledge, the only natural problem whose state of the art complexity is lower than the general complexity for solving the \lovasz local lemma but still larger than its doubly logarithmic lower bound \cite{GHKM21,GHKM21}.  Thus it is located in the interior of the unknown region. This role in the  complexity landscape provides further motivation for understanding the $\Delta$-coloring problem and its randomized complexity in detail.

\smallskip

In general, it is well-known that the randomized complexity and the deterministic complexity of local distributed graph problems are extremely interconnected and that this is even necessary \cite{Chang2016a}.
Many local graph problems are solved via the shattering method, which first uses a randomized algorithm to solve the problem in large parts of the graph such that unsolved parts consists of exponentially smaller components. On each of these, one can then run the best deterministic algorithm resulting in exponentially faster randomized algorithms.
We also leverage our deterministic algorithm by using it in a shattering framework and obtain the following result.

\begin{tcolorbox}
	\textbf{Main Contribution II:} There is a randomized $O(\log\log n)$-round distributed algorithm that w.h.p.\  $\Delta$-colors constant-degree dense graphs (see \Cref{thm:DeltaColoringRandomized}).
\end{tcolorbox}

Similar to our deterministic algorithm this matches the lower bound on constant-degree graphs \cite{LLL_lowerbound}.
On general graphs it has complexity $\min\{\widetilde{O}(\log^{5/3}\log n), O(\Delta+\log\log n)\}$.

\paragraph{What are dense graphs and why are they challenging for the $\Delta$-coloring problem?}
All recent distributed graph coloring algorithms (but also in models like streaming, massively parallel computation, etc.) rely on decomposing the input graph into a set of sparse vertices and the set of dense vertices, e.g., \cite{hsinhao_coloring,FHM23,HKMT21,CLP20,HKNT22,HM24,ACK19,AKM22}. As explained earlier, a vertex $v$ is sparse if its neighborhood contains many non-adjacent vertices. Sparse vertices are relatively easy to deal with in randomized $\Delta$-coloring algorithms. It is well-known and widely used, see e.g., \cite{EPS15} and all of the aforementioned papers, that the simplest one-round coloring algorithm---in this algorithm each node picks a random candidate color and permanently adopts the color if no neighbors tries the same color---is likely to produce permanent slack for sparse vertices. A node receives \emph{permanent slack} if two of its neighbors are colored with the same color, as it loses only one color of its palette of available colors but two competitors for these colors. Hence, coloring a node with permanent slack brings us back to the greedy regime, i.e., we obtain a problem whose flavor is similar to the one of $\Delta+1$-coloring. This slack generation for sparse graphs has been leveraged for $\Delta$-coloring in \cite{FHM23,HM24,AKM22}, but also to obtain very efficient $\Delta+1$-coloring algorithms in various papers before. Intuitively, slack also helps when coloring with $>\Delta$ colors as a larger gap between the number of available colors and competitors for these colors increases the probability to be colored in subsequent iterations, also see \cite{SW10}.
In all of these works the dense parts of the graph are much more challenging to deal with as dependencies are larger and it is much more difficult to produce slack. On the positive side dense parts of the graph provide additional structure, i.e., dense parts contain many clique-like substructures.  Our algorithm uses this structure to produce slack, even though this requires careful coordination between the nodes in the network.

\subsection{Technical Overview \& Related Work.}
\label{sec:tecOverview}

\paragraph{Prior deterministic algorithms. } The $\Delta$-coloring problem is easy in areas of the graph where there is a node with degree $<\Delta$. We can stall coloring that node, providing temporary slack to its neighbors as the number of competitors for colors reduces due to the stalled neighbor. By iterating this stalling procedure, we obtain multiple layers around the original node, where each layer consists of vertices at the same distance from the origin. Coloring is then performed layer by layer, starting from the outermost layer and progressing inward to the node with degree $<\Delta$.
Prior deterministic algorithms leverage this idea and use the concept of so-called degree choosable components (DCC) \cite{erdos79choosability,vizing76vertex,PS95,GHKM21}.  A \emph{DCC} is a subgraph $H$ that can existentially be $\Delta$-colored, regardless of how the vertices outside of $H$ have been colored. The crucial structural theorem forming the base for prior results claims that any node is contained in a DCC of at most logarithmic diameter. On the positive side a  DCC of logarithmic diameter can be solved by brute-force in logaritmic time in the \LOCAL model, regardless of how we have colored other vertices. The main issue is that DCCs for different vertices are neighboring and may even overlap. Thus they cannot be dealt with simultaneously. Hence, some form of symmetry-breaking between the DCCs is required. As symmetry-breaking usually requires at least $\Omega(\log^* n)$ time, and DCCs have a logarithmic diameter, this approach can probably, even in the best case, only lead to algorithms with runtime $O(\log n\cdot \log^* n)$.\footnote{Just before the PODC deadline the list of accepted papers of STOC 2025 was released. It contains a paper on Distributed $\Delta$-coloring exhibiting the explained behavior of being stuck at $O(\log n\cdot \log^* n)$ even for $\Delta=O(1)$ \cite{personalCommunication}.}
Due to this observation Balliu, Brandt, Kuhn, and Olivetti even explicitly asked for a genuinely different approach to attack the $\Delta$-coloring problem \cite{BBKO2021hideandseek}.

\subsubsection*{Our approach.}
Our approach is inspired by the recent development in randomized algorithms for the problem, namely \cite{FHM23,HM24}. In both of these the focus is on locally producing permanent slack as explained for the case of sparse nodes in the previous section.

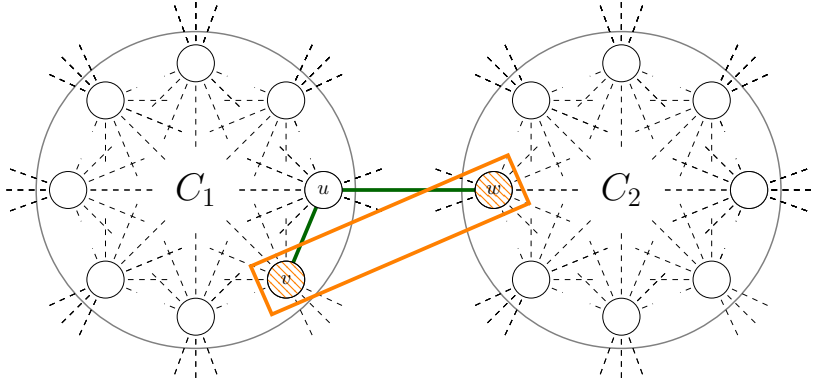
\begin{figure}
	\centering
	\scalebox{0.7}{\begin{tikzpicture}
	\def\xO{0}
	\def\yO{0}
	\def\R{3} 
	\def\Ri{2.4}
	\def\rA{0.35}
	\def\N{8} 
	
	\begin{scope}[shift={(-4,0)}]
		\draw[thick, gray] (\xO,\yO) circle(\R cm);
		
		\foreach \i in {0,1,...,\N} {
			\pgfmathsetmacro{\theta}{360/\N*\i} 
			\pgfmathsetmacro{\x}{\Ri*cos(\theta)}
			\pgfmathsetmacro{\y}{\Ri*sin(\theta)}
			\coordinate (P\i) at (\x,\y);
		}
		
		\foreach \i in {0,1,...,\N} {
			\foreach \j in {0,1,...,\N} {
				\ifnum \i<\j
				\draw[dashed] (P\i) -- ($(P\i)!0.35!(P\j)$);
				\draw[dashed] (P\j) -- ($(P\j)!0.35!(P\i)$);
                \draw[dashed] (P\i) -- ++(360/\N*\i-20:1.2cm);
                \draw[dashed] (P\i) -- ++(360/\N*\i+20:1.2cm);
                \draw[dashed] (P\i) -- ++(360/\N*\i:1.2cm);
				\fi
			}
		}
	\end{scope}
	
	\begin{scope}[shift={(4,0)}]
		\draw[thick, gray] (\xO,\yO) circle(\R cm);
		
		\foreach \i in {0,1,...,\N} {
			\pgfmathsetmacro{\theta}{360/\N*\i} 
			\pgfmathsetmacro{\x}{\Ri*cos(\theta)}
			\pgfmathsetmacro{\y}{\Ri*sin(\theta)}
			\coordinate (Q\i) at (\x,\y);
		}
		
		\foreach \i in {0,1,...,\N} {
			\foreach \j in {0,1,...,\N} {
				\ifnum \i<\j
				\draw[dashed] (Q\i) -- ($(Q\i)!0.35!(Q\j)$);
				\draw[dashed] (Q\j) -- ($(Q\j)!0.35!(Q\i)$);
                \draw[dashed] (Q\i) -- ++(360/\N*\i-20:1.2cm);
                \draw[dashed] (Q\i) -- ++(360/\N*\i+20:1.2cm);
                \draw[dashed] (Q\i) -- ++(360/\N*\i:1.2cm);
				\fi
			}
		}
	\end{scope}
	
	\draw[] (P0) -- (Q4);
	
	\draw[color=black!60!green, line width = 2px] (P7) -- (P0) -- (Q4);
	
	\foreach \i in {0,1,...,\N} {
		\draw[fill=white] (P\i) circle(\rA cm);
	}
    \draw[pattern=north west lines, pattern color=orange] (P7) circle(\rA cm);
	\foreach \i in {0,1,...,\N} {
		\draw[fill=white] (Q\i) circle(\rA cm);
	}
	\draw[pattern=north west lines, pattern color=orange] (Q4) circle(\rA cm);
    \node[] at (P7) {$v$};
    \node[] at (P0) {$u$};
    \node[] at (Q4) {$w$};
    \node[] at (-4,0) {\huge$C_1$};
    \node[] at (4,0) {\huge$C_2$};


    \coordinate (A) at (42,42);
    \edef\angleX{23}
    \newdimen\pax
    \pgfextractx{\pax}{\pgfpointanchor{P7}{center}}
    \newdimen\pay
    \pgfextracty{\pay}{\pgfpointanchor{P7}{center}}
    \newdimen\qax
    \pgfextractx{\qax}{\pgfpointanchor{Q4}{center}}
    \newdimen\qay
    \pgfextracty{\qay}{\pgfpointanchor{Q4}{center}}

    \pgfmathparse{90 - atan((\qax-\pax) / (\qay-\pay))} \let\angleAB\pgfmathresult
    
    \node[draw,orange,fit=(P7) (Q4), rotate fit = \angleAB, text width =170pt, text height = 5pt, inner sep = -10pt, line width = 2px] {};
\end{tikzpicture}}
	\caption{For each hard clique, such as $C_1$ and $C_2$, we identify slack triads—for example, $(u, v, w)$ in $C_1$. The slack pair vertices, $v$ and $w$, are then same-colored (see \Cref{fig:slackTriadVirtual}), providing slack to vertex $u$.}
	\label{fig:slacktriad}
\end{figure}
\textbf{Slack triads (\Cref{fig:slacktriad}).}
To achieve this, we aim to find many so-called slack triads\footnote{A similar yet slightly different concept first appeared in a prior randomized algorithm under the name of a $T$-node \cite{GHKM21}.}.  A slack triad is a triple of nodes $u$, $v,$ and $w$ such that $v$ and $w$ are both neighbors of $u$, but $v$ and $w$ are non-adjacent. In that case, we can same-color $v$ and $w$ providing slack to $u$. We call $u$ the \emph{slack vertex} and $v$ and $w$ the \emph{slack pair} vertices. Using a similar layering approach as for nodes with degree $<\Delta$, we can color all nodes reaching $u$ via a short (optimally constant-hop) path of uncolored vertices. The main obstacle is that slack triads should be
\begin{center}
	\emph{
		i) non-overlapping,
		ii) easily be same-colorable, and
		iii) be locally everywhere in the graph.
	}
\end{center}
The first property is important as the slack at $u$ is useless if $u$ appeared in some other slack triad where $u$ is actually colored. Also, if $v$ and $w$ appear as slack pair vertices in other slack triads we may end up having to same-color large subsets of vertices. In order to satisfy condition $ii)$ we ensure that the conflict graph of same-coloring the respective vertices of all slack triads has maximum degree $\Delta-1$. Note that this is non-trivial as potentially both slack pair vertices could have up to $\Delta-1$ neighboring conflicting slack triads, resulting in a maximum degree of $2\Delta-2\gg\Delta-1$.  See \Cref{fig:slackTriadVirtual} for an illustrations of the virtual conflict graph formed by slack triads. Property iii) is clearly needed to color all vertices of the graph.

\begin{figure}
	\centering
	\scalebox{0.65}{\begin{tikzpicture}
    \def\R{0.35}
    \def\dist{2.5}
    \def\offset{25}
    \def\xa{-3.5}
	\def\ya{0}
    \def\xb{8}
	\def\yb{-3.5}
    \def\xc{4.5}
	\def\yc{-2.5}
    \def\xd{-8}
	\def\yd{-3.5}
    \def\colorarray{{"green","red","blue","red"}}
    \def\colora{red}
    \def\colorb{green}
    \def\colorc{blue}
    \def\colord{red}

    \def\rotationa{-35}
    \def\rotationb{110}
    \def\rotationc{200}
    \def\rotationd{37}

    \coordinate (P0) at (\xa,\ya);
    \coordinate (P1) at (\xb,\yb);
    \coordinate (P2) at (\xc,\yc);
    \coordinate (P3) at (\xd,\yd);

    \coordinate (P00) at ({\xa + \dist*cos(-\offset+\rotationa)}, {\ya + \dist*sin(-\offset+\rotationa)});
    \coordinate (P01) at ({\xa + \dist*cos(\offset+\rotationa)}, {\ya + \dist*sin(\offset+\rotationa)});
    \coordinate (P10) at ({\xb + \dist*cos(-\offset+\rotationb)}, {\yb +\dist*sin(-\offset+\rotationb)});
    \coordinate (P11) at ({\xb + \dist*cos(\offset+\rotationb)}, {\yb +\dist*sin(\offset+\rotationb)});
    \coordinate (P20) at ({\xc + \dist*cos(-\offset+\rotationc)}, {\yc +\dist*sin(-\offset+\rotationc)});
    \coordinate (P21) at ({\xc + \dist*cos(\offset+\rotationc)}, {\yc +\dist*sin(\offset+\rotationc)});
    \coordinate (P30) at ({\xd + \dist*cos(-\offset+\rotationd)}, {\yd +\dist*sin(-\offset+\rotationd)});
    \coordinate (P31) at ({\xd + \dist*cos(\offset+\rotationd)}, {\yd +\dist*sin(\offset+\rotationd)});
    
    \foreach \i in {0,1,2,3}{
        \newdimen\pax
        \pgfextractx{\pax}{\pgfpointanchor{P\i0}{center}}
        \newdimen\pay
        \pgfextracty{\pay}{\pgfpointanchor{P\i0}{center}}
        \newdimen\qax
        \pgfextractx{\qax}{\pgfpointanchor{P\i1}{center}}
        \newdimen\qay
        \pgfextracty{\qay}{\pgfpointanchor{P\i1}{center}}

        \coordinate (M\i) at ({(\pax + \qax) / 2}, {(\pay + \qay) / 2});
    }
    
    \draw[color=orange, line width = 2px] (P11) -- (P20);
    \draw[color=orange, line width = 2px] (P10) -- (P01);
    \draw[color=orange, line width = 2px] (P20) -- (P01);
    \draw[color=orange, line width = 2px] (P00) -- (P21);
    \draw[color=orange, line width = 2px] (P21) -- (P30);
    \draw[color=orange, line width = 2px] (P31) -- (P00);
    \draw[color=orange, line width = 2px] (P31) -- ++(200:2.5cm);
    \draw[color=orange, line width = 2px] (P21) -- ++(350:2.5cm);
    \draw[color=orange, line width = 2px] (P10) -- ++(300:2.5cm);
    \foreach \i in {0,1,2,3}{
        
        \newdimen\pax
        \pgfextractx{\pax}{\pgfpointanchor{P\i0}{center}}
        \newdimen\pay
        \pgfextracty{\pay}{\pgfpointanchor{P\i0}{center}}
        \newdimen\qax
        \pgfextractx{\qax}{\pgfpointanchor{P\i1}{center}}
        \newdimen\qay
        \pgfextracty{\qay}{\pgfpointanchor{P\i1}{center}}

        \pgfmathparse{90 - atan((\qax-\pax) / (\qay-\pay))} \let\angleAB\pgfmathresult

        \node[fill=white,fit=(P\i0) (P\i1), rotate = \angleAB, minimum width =115pt, minimum height = 35pt, inner sep = -100pt, line width = 2px] {};
    }

    \draw[dashed, line width = 1px] (P11) -- (P20);
    \draw[dashed, line width = 1px] (P10) -- (P01);
    \draw[dashed, line width = 1px] (P20) -- (P01);
    \draw[dashed, line width = 1px] (P00) -- (P21);
    \draw[dashed, line width = 1px] (P21) -- (P30);
    \draw[dashed, line width = 1px] (P31) -- (P00);
        
    \foreach \i in {0,1,2,3}{
        \newdimen\pax
        \pgfextractx{\pax}{\pgfpointanchor{P\i0}{center}}
        \newdimen\pay
        \pgfextracty{\pay}{\pgfpointanchor{P\i0}{center}}
        \newdimen\qax
        \pgfextractx{\qax}{\pgfpointanchor{P\i1}{center}}
        \newdimen\qay
        \pgfextracty{\qay}{\pgfpointanchor{P\i1}{center}}

        \pgfmathparse{90 - atan((\qax-\pax) / (\qay-\pay))} \let\angleAB\pgfmathresult
        \draw[] (P\i) -- (P\i0);
        \draw[] (P\i) -- (P\i1);

        \node[draw,orange,fit=(P\i0) (P\i1), rotate = \angleAB, minimum width =105pt, minimum height = 35pt, inner sep = -100pt, line width = 2px, fill = orange,fill opacity = 0.2] {};
        
        \pgfmathparse{\colorarray[\i]} 
        \edef\colorname{\pgfmathresult}
        \draw[fill=white] (P\i) circle(\R cm);
        \draw[pattern = checkerboard light gray,opacity = 0.3] (P\i) circle(\R cm);
        \draw[fill=white] (P\i0) circle(\R cm);
        \draw[pattern=north west lines, pattern color=\colorname] (P\i0) circle(\R cm);
        \draw[fill=white] (P\i1) circle(\R cm);
        \draw[pattern=north west lines, pattern color=\colorname] (P\i1) circle(\R cm);

    }
    
\end{tikzpicture}
    }
	\caption{For coloring the hard cliques in the graph, we identify slack pair vertices (orange boxes), which provide the slack vertices (checkboard) with one unit of slack. To color the slack pair vertices, we construct a virtual graph where each slack pair (orange box) is treated as a single vertex. Edges between these vertices (orange) exist if there is an underlying edge in the original graph (black).}
		\label{fig:slackTriadVirtual}
\end{figure}
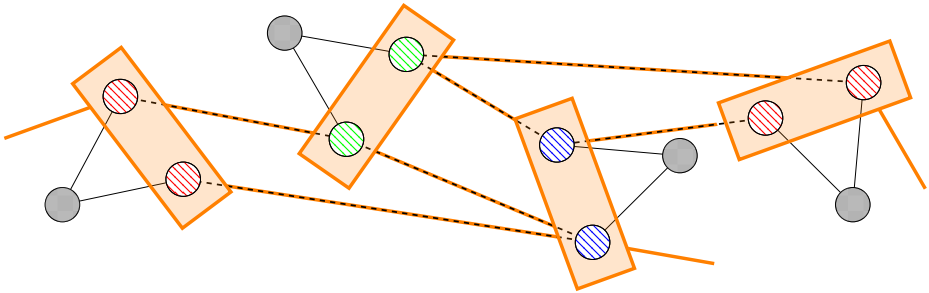

\textbf{Finding slack triads in extremely dense graphs.} The details of our algorithm for finding a set of slack triads satisfying properties i)--iii) are involved. For illustration purposes let us focus on some locally really dense part of the graph that only consists of  adjacent cliques of size $\Delta$. Fix one such clique $C$. As there are no $\Delta+1$ cliques in the graph each vertex $u$ of $C$ as one external neighbor $v_u\notin C$. For this exposition also assume that each node in $C$ has a distinct external neighbor.
The high level idea is to find two edges $(u_1,v_{u_1})$ and $(u_2,v_{{u_2}})$ leaving the clique that are not claimed by the clique on the other side. If we find these edges, nodes $u_1$ and $v_{u_1}$ and $u_2$ form a slack triad, as $u_1$ is adjacent to $v_{u_1}$ and $u_2$, and $v_{u_1}$ and $u_2$ are non-adjacent.

The \emph{sinkless orientation problem} has distributed complexity $\Theta(\log n)$ and asks to orient the edges of a graph such that every node with degree at least $3$ has one outgoing edge. By splitting each vertex into two virtual vertices, one can also ensure that every node with degree $6$ has at least two outgoing edges. If we now split each clique into two virtual \emph{vertices} and orient the intra-clique edges, we can achieve that each clique has found two edges (its outgoing edges) that are not claimed by the clique at the other end and we can form a slack triad with slack pair vertices $v_{u_1}$ and $u_2$.
This construction satisfies property $i)$ and $iii)$. But we may entirely violate property $ii)$ as all the other edges incident to the respective cliques of $v$ and $w$ may be \emph{claimed} by other cliques and hence it may be that the entire cliques of $v_{u_1}$ and $u_2$ are used in other slack triads. In order to ensure property $ii)$ we thus  use another logarithmic-time algorithm to sparsify the set of participating  intra-cluster edges. This completes the high level procedure of finding slack triads in the case where the graph locally only consists of $\Delta$ cliques.

\paragraph{Finding slack triads in dense graphs.} When the graph is less dense, each clique can have up to $\Omega(\Delta)$ outgoing edges. If we would mimic the same aproach, it may be that one vertex has multiple outgoing and multiple incoming edges. In that case it is unclear whether it should take part in the slack triad of its own clique or the slack triad of a neighboring clique. The first idea is to compute a maximal matching of intra-cluster edges and restrict the previous algorithm to the edges in the matching. The benefit is that every node has only one incident intra-cluster edge and the same reasoning as before works fine. The challenge is that a maximal matching is a greedy-type problem and hence it does not satisfy the locally everywhereness of property iii), that is, we may have parts in the graph that have no matching edge. Now, if we would simply proceed we may run into a situation where vertices cannot be colored in the sequel, not even existentially. Hence, in our full solution, we re-balance the matching to also achieve property $iii)$. To that end we develop a proposal algorithm in which cliques propose to grab intra-cluster edges. We manage proposals such that each intra-cluster edge only receives few proposals, but each clique sends many of them. Hence, on average each clique can grab some edges. We indeed ensure that \emph{every} clique grabs sufficiently many edges by reducing this re-balancing problem to the hyperedge grabbing problem of \cite{BMNSU25}.
This overview omits several technical details and in our full algorithm we also deal with parts of the graph that do not form proper cliques.

\paragraph{Why does our approach not (yet) extent to sparse graphs?}
In principle the approach of finding slack triads seems to be a promising candidate to extend to sparse graphs, and the same is true for the general scheme all of our subroutines. Instead of applying it to cliques, one may apply it to constant-radius subgraphs that contain enough outgoing edges (to send a large number of proposals). The main point where our technique does not extend is that we are unable to glue two edges grabbed by such a part together to form a slack triad. While we still think that our general approach is fruitful, a full solution should probably be less focused on edges but rather on higher order objects such as hyperedges. Ironically, as already mentioned earlier, for large $\Delta$ sparse parts are extremely simple for randomized algorithms as slack is produced everywhere with high probability with a one round color trial procedure.

\paragraph{Outline} In \Cref{sec:preliminaries}, we introduce dense graphs and present subroutines from prior work. In \Cref{sec:deterministicDeltaColoring}, we present our deterministic $\Delta$-coloring algorithm along with all essential proofs. \Cref{sec:randomizedalg} presents our randomized $\Delta$-coloring algorithm. \Cref{sec:relatedWork} contains further related work, \Cref{sec:missingProofs} contains deferred proofs, and \Cref{sec:subroutines} complementary subroutines.

\section{Preliminaries: ACD, Dense Graphs, and Hyperedge Grabbing} \label{sec:preliminaries}
\paragraph{Notation.} Given a graph $G=(V,E)$ and a vertex $v\in V$, let  $N(v)$ denote the set of $v$'s neighbors. A clique or an almost clique $C$ is a subgraph of $G$ with additional properties specified below. Let $V(C)$ represent the set of vertices contained in $C$. Sometimes, we also simply write $v\in C$ to depict a vertex in a given clique $C$. Given a set of cliques $\mathcal{C}$, we define $V(\mathcal{C}) = \bigcup_{C \in \mathcal{C}} V(C)$ as the union of the vertex sets of all cliques in $\mathcal{C}$. Furthermore, let $\deg(C)=|\{\{u,v\}\in E \mid u\in C, v\notin C\}|$ denote the degree of a clique $C$, defined as the number of outgoing edges of the clique. Throughout let $\Delta=\max_{v\in V}|N(v)|$ be the maximum degree of the input graph $G$. For an integer $k$ let $[k]$ denote the set $\{0,\ldots,k-1\}$.
Let $v$ be a node with available color palette $P(v)$ in a subgraph $\tilde{G}\subseteq G$. The \emph{slack} of $v$ in $\tilde{G}$ is $|P(v)|-d$, where $d$ is the number of uncolored neighbors of $v$ in $\tilde{G}$. A vertex has \emph{slack} if its slack is positive.

\medskip

We next explain the necessary basics of the almost clique decomposition that we require for our theorem. In the following $0<\eta, \eps<1$ are sufficiently small constants.
Two adjacent nodes $u$ and $v$ are \emph{friends}, if $|N(v)\cap N(u)|\geq (1-\eta)\Delta$. A vertex is \emph{$\eta$-dense} if it has at least $(1-\eta)\Delta$ friend neighbors, otherwise it is \emph{$\eta$-sparse}.
\begin{claim}[\cite{ACK19}]
	Any $\eta$-sparse vertex has at most $(1-\eta^2)\binom{\Delta}{2}$ edges in its neighborhood.
\end{claim}
Based on these definitions, any graph can be decomposed into sparse and dense vertices, which are further partitioned into \emph{almost cliques (ACs)}. Different coloring algorithms provide slightly varying guarantees for these ACs. In the basic version of the decomposition, each connected component of the graph induced by dense vertices and friend edges forms an AC \cite{HSS18,ACK19}. This decomposition can be computed in $O(1)$ rounds in the \LOCAL model, the ACs have diameter two, and each AC is of size $(1\pm O(\eta))\Delta$ \cite{HSS18,ACK19}. We use the following  ACD that can be obtained via a deterministic $O(1)$-round postprocessing of the basic ACD \cite{FHM23,HM24}:

\begin{restatable}[ACD computation \cite{HSS18,AKM22,FHM23,HM24}]{lemma}{lemACDcomputation}
	\label{lem:acd}
	Let $\eps=\epsvalue$.
	For any graph $G=(V,E)$, there is a partition (\emph{almost-clique decomposition (ACD)} of $V$ into sets $\Vsparse$ and $C_1, C_2, \ldots, C_t$ such that each node in $\Vsparse$ is $\Omega(\epsilon^2\Delta)$-sparse\ and for every  $i\in [t]$,
	\begin{compactenum}
		\renewcommand{\labelenumi}{(\roman{enumi})}
		\item $(1 - \eps/4)\Delta \le |C_i|\le (1+\eps)\Delta$\ ,
		\item Each $v\in C_i$ has at least $(1-\eps)\Delta$ neighbors in $C_i$:   $|N(v)\cap C_i|\ge (1-\eps)\Delta$\ ,
		\item Each node $u \not\in C_i$ has at most $(1-\eps/2)\Delta$ neighbors in $C_i$: $|N(u)\cap C_i|\le (1-\eps/2)\Delta$.
	\end{compactenum}
	Further, there is a deterministic $O(1)$-round \LOCAL algorithm to compute a valid ACD.
\end{restatable}

\begin{observation}\label{obs:outdegree}
	Any vertex $v$ of an AC $C$ has at most $\eps\Delta$ neighbors outside of $C$.
\end{observation}

\begin{definition}[Dense graph]
	We call a graph \emph{dense} the ACD computation executed with that $\eps=1/63$ does not contain any nodes that are classified as sparse.
\end{definition}
While $\Delta$-coloring algorithms  (\Cref{thm:DeltaColoring,thm:DeltaColoringRandomized}) work for any  $\Delta$, for $\Delta<63$ any dense graph can only consists of isolated cliques due to the choice of $\eps$. We believe that the parameter $\eps$ can be made larger classifying more graphs as dense at the cost of  complicating the algorithm.

\paragraph{Hyperedge Grabbing (HEG).} In a hypergraph, the \emph{minimum degree} $\delta$ is the smallest number of hyperedges incident to any vertex, while the \emph{(maximum) rank} $r$ is the largest number of vertices contained in any hyperedge. Given a hypergraph with rank $r$ and minimum degree $\delta$, the \emph{hyperedge grabbing problem (HEG)} asks each vertex to grab one of its incident edges such that no hyeredge is grabbed by more than one vertex. The problem is equivalent to the hypergraph sinkless orientation problem from \cite{BMNSU25} and can be solved in logarithmic time with their algorithm if the hypergraph's rank is sufficiently smaller than the minimum degree.
\begin{lemma}[\cite{BMNSU25}]\label{theorem:HEG}
	There is a deterministic $O(\log_{\frac{\delta}{r}}n)$-round algorithm for computing an HEG in any $n$-node multihypergraph with minimum degree $\delta$ and maximum rank $r<\delta$.
\end{lemma}

\section{\texorpdfstring{Deterministic $\Delta$-Coloring dense graphs}{Deterministic Delta-Coloring dense graphs}}
\label{sec:deterministicDeltaColoring}
The goal of this section is to prove the following theorem.
\begin{restatable}{theorem}{thmDeltaColoring}
	\label{thm:DeltaColoring}
	There is a deterministic $\min\{\widetilde{O}(\log^{5/3}n),O(\Delta+\log n)\}$-round \LOCAL algorithm for $\Delta$-coloring dense graphs that do not contain a $\Delta+1$ clique.
\end{restatable}

The definitions of the terms in the high-level overview in \Cref{alg:mainAlgorithm}, i.e., hard cliques, easy cliques, and loopholes,  follow in \Cref{sec:loopholesDefinitions}.
The bulk of the paper will be coloring the vertices in hard cliques in \Cref{sec:coloringHardCliques,sec:phase1,sec:phase2,sec:phase3,sec:phase4,sec:runtimehard,sec:phase4B}. The last step of coloring easy cliques and loopholes is detailed on in \Cref{sec:easyLoophole}. The short proof of \Cref{thm:DeltaColoring} combining all previous lemmas appears in \Cref{sec:proofMainTheorem}.
\begin{algorithm}[H]\caption{$\Delta$-coloring  dense graphs (high-level overview)}
	\label{alg:mainAlgorithm}
	\begin{algorithmic}[1]
		\STATE Compute an ACD (for $\eps=\epsvalue$) and form the ordered partition of the nodes.
		\STATE  Color vertices in hard cliques
		\STATE Color vertices in easy almost cliques and loopholes \label{step:coloreasycliques}
	\end{algorithmic}
\end{algorithm}
\vspace{-0.8cm}
\subsection{Definitions: Loopholes and Hard/Easy Almost-Cliques}\label{sec:loopholesDefinitions}
Any node in a graph with degree less than $\Delta$ can be greedily (in a centralized setting) colored with one of the $\Delta$ colors, regardless of how its neighbors have been colored. Thus, it forms in some sense a loophole for the $\Delta$-coloring problem. Another loophole is given by a non-clique four cycle \cite{erdos79choosability,vizing76vertex,GHKM21}.
A non-clique four cycle in a graph $G$ consists of four vertices $v_0,v_1,v_2,v_3$ that form a four cycle, i.e., $\{v_i, v_{i+1\mod 4}\}\in E(G)$ for $i=0,1,2,3$ and $G[\{v_0,v_1,v_2,v_3\}]$ does not form a clique. In fact, it is known that any non-clique even length cycle is a loophole. We summarize these loopholes in the following definition, see \Cref{fig:loophole} for an example.
\begin{definition}[Loophole] \label{def:loopholes}
	Consider a clique $C$ in a graph $G$. A \textit{loophole} is a subgraph which has one of the following shape:
	\begin{compactenum}
		\item A vertex in $C$ with a degree less than $\Delta$, or
		\item A non-clique cycle of even length
	\end{compactenum}
	A \emph{loophole vertex} is a vertex contained in some loophole.
\end{definition}
In the literature these loopholes also appear as \emph{degree choosable components}, or as $deg$-list colorable graphs.  A graph $G$ is  \emph{$deg$-list colorable} if any list coloring in which node the list of allowed vertices for each vertex $v\in V(G)$ has at least $deg(v)$ colors.
\begin{lemma}[\cite{erdos79choosability,vizing76vertex,GHKM21}] \label{lem:loophole}
	Non-clique cycles of even length are $deg$-list colorable.
\end{lemma}
For our usage, the lemma is useful, as we can compute any $\Delta$-coloring of $V\setminus L$ to a coloring of $V$ if $G[L]$ is a loophole.
In contrast to prior work based on such loopholes, we only consider loopholes of constant size which is crucial for obtaining a flat logarithmic runtime, see \Cref{sec:tecOverview}.

Due to \Cref{lem:loophole}, it will be easier to color almost cliques that intersect a loophole. The bulk of our algorithm deals with the remaining hard cliques, as defined next.
\begin{definition}[Easy Almost Cliques/Hard Cliques] \label{def:hardeasycliques}
	An almost clique is called a \textit{hard} if it does not contain any vertex belonging to a \textit{loophole} of at most $6$ vertices; otherwise, it is called \textit{easy}.
\end{definition}
Let $\hardcliques$ be the set of all hard cliques in $G$. Further, we refer $\vhard=\bigcup_{C' \in \mathcal{C}_{hard}} V(C')$ as the set of vertices that are in a hard clique and $\ehard$ as the edges between vertices in $\vhard$ whose endpoints lie in different cliques. We often omit the 'almost' for  \emph{hard almost cliques}.

\begin{restatable}{lemma}{lemmahardcliqueprops} \label{lem:hardcliqueprops}
	The following hold for each hard clique $C$:
	\begin{compactenum}
		\item $C$ forms a clique,
		\item Each vertex $v\in C$  has $e_C=\Delta - |C| + 1$  neighbors outside of $C$~,
		\item \label{lem:hardcliqueprops:notwons}There is no vertex $w\notin C$ with two neighbors in $C$~.
	\end{compactenum}
\end{restatable}
The proof of \Cref{lem:hardcliqueprops} is a direct consequence of the definition of hard cliques.  For example if Part~3 would not hold, one could identify a loophole in the clique as depicted in \Cref{fig:loophole}.

\subsection{Coloring hard cliques: High level overview and Pseudocode (Algorithm~\ref{alg:coloring})} \label{sec:coloringHardCliques}

\paragraph{High level overview (pseudocode in \Cref{alg:coloring}).}
Hard cliques, introduced in  \Cref{def:hardeasycliques}, are cliques that do not intersect any loophole and therefore lack access to an existing slack source. To address this we identify  a collection of so-called non-overlapping \emph{slack triads}, see \Cref{def:slacktriad} below. A slack triad for a hard clique $C$ consists of a vertex $u\in C$ and two non-adjacent neighbors $v\in C$ and $w\notin C$ of $u$. Once identified these nodes, we aim at same-coloring the \ \emph{slack pair} $v$ and $w$ (\myStepCref{alg:coloring:8}, \Cref{sec:phase4}), providing slack to the uncolored node $u$, effectively providing a toehold to the whole clique such that all nodes can be colored efficiently, see \myStepCref{alg:coloring:9} and \Cref{sec:phase4B}.
One challenge is to ensure that we can even effectively same-color the slack pairs. Consider the virtual graph consisting of one vertex for each such slack pair and an edge between two 'slack pair vertices' if there is an edge between any of their corresponding vertices in the underlying graph $G$. Same-coloring the vertices in each slack pair corresponds to $\Delta$-coloring this virtual graph. If $u$ and $w$ are poorly chosen the maximum degree of the virtual graph can be as large as $2\Delta-2$ and there is little hope in effectively coloring it with $\Delta$ colors. The bulk of this section (\myStepCref{alg:coloring:1} to \myStepCref{alg:coloring:7} dealt with in \Cref{sec:phase1,sec:phase2,sec:phase3}) aims at identifying these slack triads. See \Cref{sec:tecOverview} for intuition on these steps.  Phase~1 is is presented in \Cref{sec:phase1} and contains the initial maximal matching computation as well as the explained proposal algorithm to modify the matching such that conditions i)--iii) from \Cref{sec:tecOverview} are satisfied. The latter is the most involved part of our algorithm.
\medskip

Let $\mathcal{C}$ (\hardcliques) denote the set of all almost cliques (hard cliques). Define the subgraph
$$G' \coloneqq G\left[\{C \in\mathcal{C}_{hard} \mid \forall v \in V(C), N(v) \cap \vhard\neq \emptyset\}\right]$$
that consists of all hard almost cliques in $\mathcal{C}_{hard}$ where each vertex of the clique has at least one neighbor in a hard  clique. We begin with a high level overview on the algorithm. Thereafter we detail on each step and its properties in a separate section.

\begin{algorithm}
	\caption{Coloring vertices in hard cliques (high level overview)}
	\label{alg:coloring}
	\begin{algorithmic}[1]
		\Statex \hspace{-0.6cm}\textbf{Phase 1: Balanced Matching (\Cref{sec:phase1})}
		\STATE Compute Maximal Matching $F_1$ in $G'$ \label{alg:coloring:1}
		\STATE \label{alg:coloring:2}Compute a HEG on the hypergraph
		\STATE \label{alg:coloring:3}Rearrange edges in $F_1$ to vertex that grabbed the edge to obtain matching $F2$
		\STATE \label{alg:coloring:4}Orient edges in $F_2$ to the vertex that grabbed the edge.
		\Statex \hspace{-0.6cm}\textbf{Phase 2: Sparsifying the matching (\Cref{sec:phase2})}
		\STATE \label{alg:coloring:5}Apply the Degree Splitting to reduce incoming and outgoing edges per clique.
		\STATE \label{alg:coloring:6}Discard all outgoing edges per clique except two
		\Statex \hspace{-0.6cm}\textbf{Phase 3: Slack Triad Forming (\Cref{sec:phase3})}
		\STATE \label{alg:coloring:7}Define a slack triad from the two vertices with an outgoing edge per clique
		\Statex \hspace{-0.6cm}\textbf{Phase 4: Coloring (\Cref{sec:phase4,sec:phase4B})}
		\STATE Compute a coloring of the slack pairs in the graph induced by them. \label{alg:coloring:8} \label{TEST}
		\STATE \label{alg:coloring:9}Color remaining hard vertices with $O(1)$ rounds of $deg+1$-list coloring instances\label{step:coloring}
	\end{algorithmic}
\end{algorithm}

\subsection{Phase 1: Balanced Matching}\label{sec:phase1}
The first step of the coloring is to determine a matching in the graph $(\vhard,\ehard)$ that considers only edges of $E(G)$ whose endpoints lie in different hard cliques.

The main part of this stage is to rearrange and orient the matching $F_1$ such that each participating clique ends up with at least $28$ incident outgoing edges.   Let  $\hardcliqueshso\subseteq\hardcliques$ be the set of all hard cliques for which each vertex is adjacent to at least one vertex in another hard clique. Define the  following multihypergraph $H$ to apply \Cref{theorem:HEG} on it.

$$H=(V_H,E_H)~.$$
\textbf{The vertex set $V_H$.} The vertex set $V_H$ consists of $28$ virtual vertices for each hard clique $C\in\hardcliqueshso$ that are obtained by partitioning $C$ into sets $\hhc_{C}^1, \ldots, \hhc_{C}^{28}$. Each of these virtual \emph{sub cliques} forms one vertex of $H$.

\paragraph{Intuition for the proposal algorithm.} Each virtual sub clique sends out one request for each of its vertices to grab a close-by edge of $F_1$; different vertices of the sub clique send requests to different edges, as otherwise loopholes would exist. Then, the hypergraph $H$ has one hyperedge for each edge of $F_1$  consisting of the sub cliques requesting to grab that edge. In a solution to the HEG problem on $H$ sub clique exclusively grabs one edge of $F_1$. In a second step this relation is used to modify the matching $F_1$ and to orient the resulting edges such that each sub clique obtains one outgoing edge. This totals in $28$ outgoing edges for each hard clique in $\hardcliqueshso$.

\smallskip

\noindent\textbf{The edge set $E_H$.} Formally, describing the edge set of $H$ is more involved, in particular, as we require several technical definitions  to later formally define the matching $F_2$ stemming from a solution to HEG on $H$. To this end, define  $f:V(\hardcliqueshso)\rightarrow V(F_1)$:
\[f(v)=
	\begin{cases}
		v & \text{if $v$ has incident edge in $F_1$}                                            \\
		u & \text{otherwise, where $u$ is the vertex with minimum ID in $N_{out}(v)\cap\vhard$}
	\end{cases}
\]
The function $f(v)$ is well-defined because every vertex $v$ belongs to one of the participating cliques. This ensures that $v$ either has an incident edge in $F_1$ or a neighbor in $N_{out}(v) \cap \vhard$.
Note that if $f(v)=u$ with $u\neq v$, then $u$ has an incident edge in $F_1$ as $F_1$ is a maximal matching between hard vertices and both $v$ and $u$ are hard vertices. Let $\phi(v)$  be the unique edge of $F_1$ incident to $f(v)$. For a vertex $v\in V(\hardcliqueshso)$ let $Q_i(v)\in V_H$ denote the sub clique that contains $v$. The hypergraph $H$ has one hyperedge $f_e=\{Q_{i}(v) \mid v\in V, \phi(v)=e\}$ for each edge $e\in F_1$ and $f_e$ contains exactly those sub cliques that request to grab edge $e$.

\smallskip

The following lemma shows that all vertices of a sub clique request to grab distinct sets of edges. It is proven  by contradiction: If two vertices of a sub clique would request to grab the same edge, one can construct a loophole intersecting their hard clique.

\begin{restatable}{lemma}{lemmasetsize}\label{lem:TrueProposal}
	For each $C\in\hardcliqueshso, i\in [28]$ we have $|f(V(Q_{C}^i))|=|\phi(V(Q_{C}^i))|=|Q_{C}^i|$.
\end{restatable}
Let $r_H=\max_{e\in F_1}|f_e|$ be the rank of the hypergraph $H$ and let $\delta_H=\min_{Q\in V_H}|\{f\in E_H \mid Q\in f\}|$ be the minimum degree of $H$. In order to apply \Cref{theorem:HEG} for computing an HEG of $H$, we show that the rank $r_H$ is smaller than the minimum degree $\delta_H$.
\begin{lemma} \label{lem:rankMinDegree}
	$\delta_H>1.1 r_H$ holds.
\end{lemma}
\begin{proof}
	First, we analyze the number of proposals received by an edge $e \in F_2$, where $e$ is represented by a hyperedge $f_e$ in $H$, corresponding to the rank $r$ of $H$. The edge $e$ receives proposals from its own endpoints as well as from (some) neighboring vertices that are adjacent to both endpoints of $e$. Each endpoint of $e$ is adjacent to at most $\eps\Delta$ vertices outside its own clique due to \Cref{obs:outdegree}. Consequently, the total number of proposals received by $e$ is bounded by $2\eps\Delta$. This also represents the maximum rank $r$ in $H$.

	The degree $\delta_H$ in $H$ is represented by the number of proposals each vertex sends in $H$. By definition each vertex of $H$ represent a sub clique $Q_C^i$ and each such sub clique sends exactly one proposal for each of its vertices. By \Cref{lem:TrueProposal} we know that each vertex proposes to a different edge. Therefore the number of send proposals per vertex in $H$ and $\delta_H$ is at least $\lfloor 1/28 (\Delta-\eps\Delta+1)\rfloor \geq 1/28 (\Delta-\eps\Delta+1) - 1$. Using $\eps= \epsvalue$ at $(*)$  we obtain
	\vspace{-0.2cm}
	\begin{align}
		\label{eqn:HEG}
		\delta_H\geq 1/28 (\Delta-\eps\Delta+1) - 1>1/28 (\Delta-\eps\Delta)\stackrel{(*)}{\geq} 2.2\eps\Delta\geq 1.1r_H~. & \qedhere
	\end{align}
\end{proof}

Intuitively, \Cref{lem:rankMinDegree} states that each subclique sends out more requests to grab edges than any edge in the initial maximal matching receives. As a result, on average, each subclique can claim at least one edge exclusively. Furthermore, \Cref{theorem:HEG} guarantees that such an assignment can be computed efficiently in $O(\log_{r_H/\delta_H} n)=O(\log n)$ rounds.

\paragraph{Computing $F_2$ from a solution to HEG on $H$.}
Given the matching $F_1$ obtained in \myStepCref{alg:coloring:1} of  \Cref{alg:coloring}, along with the result of the subsequent HEG instance, we now demonstrate the method for constructing a matching $F_2$ together with an orientation of its edges ensuring that each clique is incident to at least 28 outgoing edges from $F_2$.

We obtain the set $F_2$ by starting with the empty set and by adding edges to $F_2$ via performing the following process for each edge  $e\in F_1$ and corresponding hyperedge $f_e\in E_H$. Let $\phi''(f_e)\in V_H\cup \{\bot\}$ be the sub clique that grabbed the hyperedge $f_e$ based on the result of the HEG on $H$ obtained via \Cref{theorem:HEG}. If $\phi''(f_e)\neq \bot$. let $v_e=\phi'(f_e)\in V(\hardcliqueshso)$ be the unique vertex of the sub clique that sent a request to edge $e$.
For each edge $e\in F_1$ we add the edge $\{v_e, f(v_e)\}$ to $F_2$ and orient it to  $f(v_e)$. See \Cref{fig:flipping} for an illustration of this process.

\begin{figure}
	\centering
	\begin{subfigure}{.48\textwidth}\label{fig:flipping:before}
		\centering
		\scalebox{0.55}{\begin{tikzpicture}
    \clip (-6.5,2) rectangle + (13,-9);
    \def\xO{0}
    \def\yO{0}
    \def\R{3} 
    \def\Ri{2.4}
    \def\rA{0.35}
    \def\N{8} 

    \begin{scope}[shift={(-3.5,0)}]
        \foreach \i in {15,345,285,315} {
            \pgfmathsetmacro{\x}{\Ri*cos(\i)}
            \pgfmathsetmacro{\y}{\Ri*sin(\i)}
            \coordinate (P\i) at (\x,\y);
        }
    \end{scope}

    \begin{scope}[shift={(3.5,0)}]
        \foreach \i in {165,195,225,255} {
            \pgfmathsetmacro{\x}{\Ri*cos(\i)}
            \pgfmathsetmacro{\y}{\Ri*sin(\i)}
            \coordinate (Q\i) at (\x,\y);
        }
    \end{scope}

    \begin{scope}[shift={(0,-{(7*0.86})}]
        \foreach \i in {45,75,105,135} {
            \pgfmathsetmacro{\x}{\Ri*cos(\i)}
            \pgfmathsetmacro{\y}{\Ri*sin(\i)}
            \coordinate (R\i) at (\x,\y);
        }
    \end{scope}

    \draw[color=red, line width = 2px] (P285) -- (R135);
    \draw[] (R75) -- (Q225);
    \draw[color=red, line width = 2px] (P345) -- (Q225);

    \draw[color = blue,>-Stealth,shorten >=0.28cm + \pgflinewidth, line width = 2px] (R75) to [out=105,in=200] (Q225);
    \draw[color = blue,>-Stealth,shorten >=0.28cm + \pgflinewidth, line width = 2px] (R135) to [out=55,in=350] (P285);
    
    \begin{scope}[shift={(-3.5,0)}]
        \draw[thick, gray] (\xO,\yO) circle(\R cm);
        \draw[dashed] (0,0) -- ({\R * cos(330)},{\R * sin(330)});
        \draw[dashed] (0,0) -- ({\R * cos(30)},{\R * sin(30)});
        \draw[dashed] (0,0) -- ({\R * cos(270)},{\R * sin(270)});
    \end{scope}
    
    \begin{scope}[shift={(3.5,0)}]
        \draw[thick, gray] (\xO,\yO) circle(\R cm);
        \draw[dashed] (0,0) -- ({\R * cos(150)},{\R * sin(150)});
        \draw[dashed] (0,0) -- ({\R * cos(210)},{\R * sin(210)});
        \draw[dashed] (0,0) -- ({\R * cos(270)},{\R * sin(270)});
    \end{scope}
    
    \begin{scope}[shift={(0,-{(7*0.86})}]
        \draw[thick, gray] (\xO,\yO) circle(\R cm);
        \draw[dashed] (0,0) -- ({\R * cos(30)},{\R * sin(30)});
        \draw[dashed] (0,0) -- ({\R * cos(90)},{\R * sin(90)});
        \draw[dashed] (0,0) -- ({\R * cos(150)},{\R * sin(150)});
        
    \end{scope}
    
    \foreach \i in {15,345,285,315} {
        \draw[fill=white] (P\i) circle(\rA cm);
    }
    \foreach \i in {165,195,225,255} {
        \draw[fill=white] (Q\i) circle(\rA cm);
    }
    \foreach \i in {45,75,105,135} {
        \draw[fill=white] (R\i) circle(\rA cm);
    }
    \node[] at (-1.5,-2.95) {\color{blue}\textbf{?}};
    \node[] at (0.4,-2.5) {\color{blue}\textbf{?}};

    \node[] at (R135) {$u$};
    \node[] at (R75) {$v$};
    \node[] at (P285) {$w$};
\end{tikzpicture}}
		\caption{The red edges indicate a partial maximal matching on a section of a graph that result after \myStepCref{alg:coloring:1}. The proposals of the vertices $u$ and $v$ are indicated by the blue arrows.}
	\end{subfigure}
	\begin{subfigure}{.48\textwidth}\label{fig:flipping:after}
		\centering
		\scalebox{0.55}{\begin{tikzpicture}
    \clip (-6.5,2) rectangle + (13,-9);
    \def\xO{0}
    \def\yO{0}
    \def\R{3} 
    \def\Ri{2.4}
    \def\rA{0.35}
    \def\N{8} 

    \begin{scope}[shift={(-3.5,0)}]
        \foreach \i in {15,345,285,315} {
            \pgfmathsetmacro{\x}{\Ri*cos(\i)}
            \pgfmathsetmacro{\y}{\Ri*sin(\i)}
            \coordinate (P\i) at (\x,\y);
        }
    \end{scope}

    \begin{scope}[shift={(3.5,0)}]
        \foreach \i in {165,195,225,255} {
            \pgfmathsetmacro{\x}{\Ri*cos(\i)}
            \pgfmathsetmacro{\y}{\Ri*sin(\i)}
            \coordinate (Q\i) at (\x,\y);
        }
    \end{scope}

    \begin{scope}[shift={(0,-{(7*0.86})}]
        \foreach \i in {45,75,105,135} {
            \pgfmathsetmacro{\x}{\Ri*cos(\i)}
            \pgfmathsetmacro{\y}{\Ri*sin(\i)}
            \coordinate (R\i) at (\x,\y);
        }
    \end{scope}

    \draw[>-Stealth,color=black!60!green, line width = 2px,shorten >=0.28cm + \pgflinewidth](R135) -- (P285);
    \draw[>-Stealth,color=black!60!green, line width = 2px,shorten >=0.28cm + \pgflinewidth] (R75) -- (Q225);
    \draw[] (P345) -- (Q225);

    \begin{scope}[shift={(-3.5,0)}]
        \draw[thick, gray] (\xO,\yO) circle(\R cm);
        \draw[dashed] (0,0) -- ({\R * cos(330)},{\R * sin(330)});
        \draw[dashed] (0,0) -- ({\R * cos(30)},{\R * sin(30)});
        \draw[dashed] (0,0) -- ({\R * cos(270)},{\R * sin(270)});
    \end{scope}
    
    \begin{scope}[shift={(3.5,0)}]
        \draw[thick, gray] (\xO,\yO) circle(\R cm);
        \draw[dashed] (0,0) -- ({\R * cos(150)},{\R * sin(150)});
        \draw[dashed] (0,0) -- ({\R * cos(210)},{\R * sin(210)});
        \draw[dashed] (0,0) -- ({\R * cos(270)},{\R * sin(270)});
    \end{scope}
    
    \begin{scope}[shift={(0,-{(7*0.86})}]
        \draw[thick, gray] (\xO,\yO) circle(\R cm);
        \draw[dashed] (0,0) -- ({\R * cos(30)},{\R * sin(30)});
        \draw[dashed] (0,0) -- ({\R * cos(90)},{\R * sin(90)});
        \draw[dashed] (0,0) -- ({\R * cos(150)},{\R * sin(150)});
    \end{scope}
    
    \foreach \i in {15,345,285,315} {
        \draw[fill=white] (P\i) circle(\rA cm);

    }
    \foreach \i in {165,195,225,255} {
        \draw[fill=white] (Q\i) circle(\rA cm);

    }
    \foreach \i in {45,75,105,135} {
        \draw[fill=white] (R\i) circle(\rA cm);

    }

    \node[] at (R135) {$u$};
    \node[] at (R75) {$v$};
    \node[] at (P285) {$w$};
\end{tikzpicture}}
		\caption{The green oriented edges show the partial matching after the HEG procedure in \myStepCref{alg:coloring:2}.\\$\ $\\}
	\end{subfigure}
	\caption{This example illustrates the intuition behind the HEG procedure in \Cref{alg:coloring}, applied to three cliques with their respective partitions. It demonstrates the process of edge flipping triggered by the HEG. Initially, vertices request to grab an edge, as indicated by the blue arrows based on the function $f$. This request can either target an adjacent matching edge (e.g., vertex $v$) or a matching edge adjacent to a neighbor (e.g., vertex $u$). As the algorithm progresses, the vertices $(u, v, w)$ form the slack triad for the bottom clique assuming $u$ becomes the slack vertex.}
	\label{fig:flipping}
\end{figure}
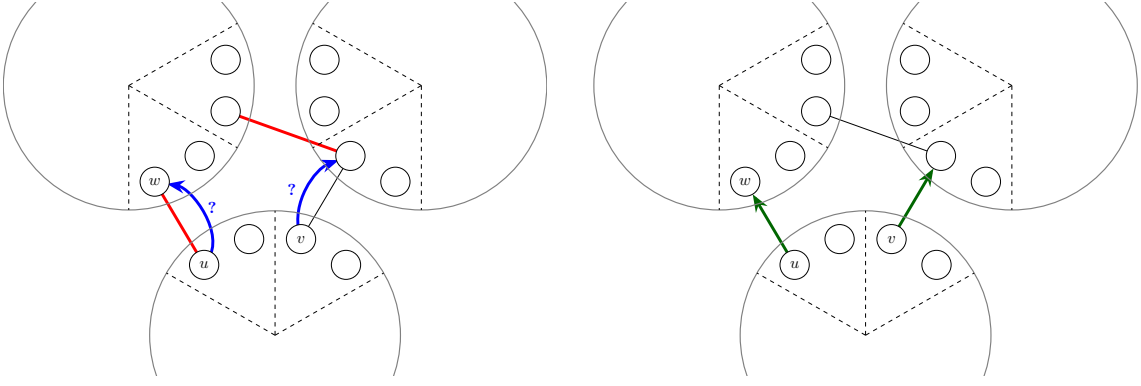
The next lemma states the properties of the resulting matching $F_2$.
\begin{lemma}[Balanced Matching] \label{lemma:twoAdjacentEdges}
	$F_2$ is a matching whose edges are directed and for each hard clique $C$ at least one of the following holds:
	\begin{compactitem}
		\item $\typeI$: $C$ has at least $28$ outgoing edges in $F_2$, or
		\item $\typeII$: $C$ has an adjacent easy AC.
	\end{compactitem}

\end{lemma}
\begin{proof}
	During the construction of $F_2$, each edge is assigned a direction based on the result of the hyperedge grabbing process, which uniquely determines the vertex grabbing that edge; recall that different vertices in a sub clique request from different edges.

	First we prove that $F_2$ is a matching.

	Consider some vertex $v$. \textit{Case 1 ($v$ does not grab any edge in $F_2$):} By the definition of $F_2$, the only way that $v$ is the endpoint of an edge in $F_2$ not grabbed by $v$ is if there was some edge $e\in F_1$ incident to $v$, edge $e$  is grabbed  by some other vertex $y\neq v$, modifying the edge to be $\{v,y\}$. But as $e$ can only be grabbed by a single vertex and $F_1$ is a matching, there can only be one edge of this type.   \textit{Case 2A ($v$ grabs one incident edge, $v$ has incident edge $e\in F_1$):}  By the definition of $f(v)$ node $v$ proposed to grab edge $e$, and hence $e$ is also the edge grabbed by $v$. As $F_1$ is a matching and edge $e$ is not grabbed by any vertex besides $v$, there cannot be any other vertex $y$ grabbing an edge incident to $v$.  \textit{Case 2B ($v$ grabs one incident edge, $v$ has no incident edge in $F_1$):} If $v$  has no incident edge in $F_1$, it can also not become adjacent to an edge in $F_2$ that $v$ does not grab.
	Lastly observe that there are no further cases as $v$ grabs at most one edge by definition.

	Consider a hard clique $C\in \hardcliqueshso$ and its  28 subcliques $\hhc_{C}^1, \ldots, \hhc_{C}^{28}$. Each  grabs at least one edge during the HEG process.  Consequently, the clique $C$ must be of $\typeI$. By definition of $\hardcliqueshso$ all other hard cliques are of $\typeII$.
\end{proof}

\subsection{Phase 2: Sparsifying the Core Matching} \label{sec:phase2}

This sparsification step aims to limit the number of each hard clique’s incident edges, essential for coloring slack pairs in \myStepCref{alg:coloring:8},  see \Cref{sec:phase4}.

We split each hard clique $C$ into two subsets using a more structured approach than in the previous step. Let $Q_C^+ \subseteq C$ denote the set of vertices that have an outgoing edge in $F_2$. The remaining vertices, $Q_C^- = C \setminus Q_C^+$, consist of those incident to an incoming edge in $F_2$ or are not incident to an edge in $F_2$ at all. Let $\mathcal{Q}^+=\{Q_C^+ \mid C\in \hardcliques\}$ and $\mathcal{Q}^-=\{Q_C^- \mid C\in \hardcliques\}$ be the sets of all such sub cliques. Construct a virtual graph $G_Q$ with vertex set  $\mathcal{Q}^+ \cup \mathcal{Q}^-$. Two vertices in $G_Q$ are connected by an edge if there exists an edge in $F_2$ with endpoints in their corresponding sub-cliques. We apply the degree-splitting algorithm from \Cref{thm:degreeSplitting} to $G_Q$ to limit the number of edges incident to each hard clique, retaining $2$ outgoing edges and imposing an upper bound on incoming edges. The outcome  is summarized in the following lemma:

\begin{lemma}[Low-Degree Balanced Matching] \label{lemma:advancedMatching} 
	Given the graph $G_H$ and the Matching $F_2$ with the properties from Lemma \ref{lemma:twoAdjacentEdges}, the second phase of Algorithm \ref{alg:coloring} computes a oriented matching $F_3\subseteq F_2$  with the following property: each hard clique $C$ has less than $\frac{1}{2}(\Delta-2\eps\Delta-1)$ incoming edges in $F_3$ and additionally one of the following holds:
	\begin{compactitem}
		\item $\typeIplus$: $C$ is adjacent to exactly $2$ outgoing edges in $F_3$, or
		\item $\typeII$: $C$ has an adjacent easy clique.
	\end{compactitem}
\end{lemma}
\begin{proof}
	We apply \Cref{cor:degsplit} with $\eps' = 1/100$, $i = 2$ on the graph $G_Q$, and partition the edges in $G_Q$ into $2^i = 2^2 = 4$ parts and we only keep the edges of $F_2$ that are in the first part. Let $F3 \subseteq F_2$ denote the resulting matching of edges, which is as subset of $F_2$ is still a matching. We also keep the initial orientation of edges. We obtain that each
	\begin{align*}
		\deg_{F3}(\mathcal{Q}^+) & \geq (\deg_{F_2}(\mathcal{Q}^+)/4-\eps'\deg_{F_2}(\mathcal{Q}^+) - 4 \geq  2                      \\
		\deg_{F3}(\mathcal{Q}^-) & \leq (\deg_{F_2}(\mathcal{Q}^-)/4+\eps'\deg_{F_2}(\mathcal{Q}^-) + 4 \leq  \Delta/4+\eps'\Delta+4 \\
		                         & \leq  \Delta/4+1\Delta/100+\Delta/5
		\leq \frac{1}{2}(\Delta-2\eps\Delta-1),
	\end{align*}
	where the first sequence of inequality holds due to $\deg_{F_2}(\mathcal{Q}^+)\geq 28$ and the last sequence of inequality holds due to $\Delta\geq\deltavalue$ and $\eps=\epsvalue$. Note that we may trivially assume $\Delta\geq 28$, as otherwise the graph consists of isolated cliques only.
\end{proof}
\subsection{Phase 3: Forming Slack Triads}\label{sec:phase3}

We use the matching $F_3$ to define a slack triad for each hard clique, thereby providing it with a source for permanent slack. Also see \Cref{fig:slacktriad}.
\begin{definition}[Slack Triad] \label{def:slacktriad}
	An ordered triple of vertices $(u,v,w)$ is a \emph{slack triad} if $v,w \in N(u)$ and $v$ and $w$ are non-adjacent. In this case we call $u$ the \emph{slack vertex} of the slack triad and $v$, $w$ form the \emph{slack pair}. We also refer $v$ and $w$ as \emph{slack pair vertices}.
\end{definition}
Given the matching $F_3$, for each hard clique $C$, we select an arbitrary tail of one of its outgoing edges, denoted as $e_1$ in $F_3$, and set it as the slack node for the corresponding slack triad of $C$. The slack triad is then completed by including the head of $e_1$ and the tail of the other outgoing edge of $C$ in $F_3$ as the slack pair vertices. Note that the slack pair vertices are indeed non-adjacent as required since otherwise there would be another vertex in $C$ that would define a loophole with the vertices of the slack triad.
\begin{lemma}[Forming slack triads] \label{lemma:step3} 
	Given the Matching $F_3$ with the properties of \Cref{lemma:advancedMatching}, in $O(1)$ rounds one can compute a collection $\mathcal{S}$ of slack triads with the following properties:
	\begin{compactenum}[\itshape i)\upshape]
		\item Each clique of $\typeIplus$ from \Cref{lemma:advancedMatching} contains a slack node of a slack triad of $\mathcal{S}$.
		\item \label{prop2}The slack triads in $\mathcal{S}$ are vertex disjoint.
		\item \label{prop3}Each clique contains at most $\frac{1}{2}(\Delta-2\eps\Delta-1)+1$
		slack pair vertices.
	\end{compactenum}
\end{lemma}
\begin{proof}
	The collection of all slack triads from $\typeIplus$ cliques, constructed as explained above, n forms the set $\mathcal{S}$. Identifying the slack triads can clearly done in $O(1)$ time.

	\textbf{Property i):} We need to prove that the process indeed forms slack triads, i.e., that the two slack pair vertices are non-adjacent. Since one vertex in the pair lies outside the clique and is already adjacent to the slack vertex, it cannot be adjacent to any other vertex in the clique of the slack vertex due to Part~\ref{lem:hardcliqueprops:notwons} in \Cref{lem:hardcliqueprops}.

	\textbf{Property ii):} Consider a slack triad $(u, v, w)$ associated with a clique $C$, and, without loss of generality, let $(u,w)$ and $(v,v')$ be the outgoing edges of $F_3$ used to construct the slack triad. As $F_3$ forms a matching, and both edges are outgoing from $C$, none of the four vertices is used in any other slack triad construction. In particular, all slack triads are disjoint.
	In the worst case, every vertex in the clique that is adjacent to an incoming edge is selected as a slack pair vertex by some other clique. By \Cref{lemma:advancedMatching}, the number of incoming edges per clique is at most $\frac{1}{2}(\Delta - 2\eps\Delta-1)$. Additionally, each clique contains at most one extra slack pair vertex stemming from the clique's own slack triad. This results in a total of at most $\frac{1}{2}(\Delta - 2\eps\Delta-1) + 1$ slack pair vertices per clique.
\end{proof}
\subsection{Phase 4A: Coloring Slack Pairs}\label{sec:phase4}
We start the coloring process with the slack pairs of the slack triads $\cS$. Since there is no edge between the two slack pair vertices one can assign both the same color providing the slack node one unit of slack without violating the coloring. Nevertheless, vertices of the slack triads can be adjacent and consequently we need to coordinate color choices between the different slack pairs. Define the following virtual graph capturing the conflict relation between slack pairs.
$$G_V=(V_V, E_V)$$
where $V_V$ contains one vertex for each slack pair in $\cS$ and there exists an edge in $E_V$ between two slack pairs if there is an edge in the original graph between any of the slack pairs' vertices.  See \Cref{fig:slackTriadVirtual} for an illustration. \Cref{lemma:triaddegree} proves that this virtual graph instance can indeed be colored via a  $\mathsf{MaxDegree}+1$-coloring procedure, when using color space $[\Delta]$.
\begin{lemma} \label{lemma:triaddegree}
	The maximum degree of $G_V$ is at most $\Delta - 2$.
\end{lemma}
\begin{proof}Consider an arbitrary slack triad $(u,v,w)$ of a clique $C$ with slack pair vertices $v$ and $w$ and slack vertex $u$. Without loss of generality let $u,v\in C$ and $w\in C'\neq C$.
	Due to \Cref{lemma:step3}, property iii) $v$ and $w$ can be adjacent to at most $\frac{1}{2}(\Delta - 2\eps\Delta-1)$ other slack pair vertices in their respective cliques.
	Additionally, due to \Cref{obs:outdegree} $v$ is adjacent to at most $\eps\Delta$ other slack pair vertices in cliques different from $C$. The same holds for node $w$, except that the external neighbor $u$ of $w$ is not a slack pair vertex. Hence $w$ has at most $\eps\Delta-1$ slack pair vertices in cliques different from $C'$.  Thus, the maximum degree of $G_V$ is at most
	\begin{align*}
		2 \cdot \frac{1}{2}(\Delta - 2\eps\Delta-1) + (\eps\Delta - 1) + \eps\Delta= \Delta-2~. & \qedhere 
	\end{align*}
\end{proof}
\subsection{Phase 4B: Coloring remaining hard vertices}\label{sec:phase4B}
\begin{lemma}[Coloring] \label{lemma:step4}
	In Step~\ref{step:coloring} all remaining vertices in hard cliques can be colored with a two $deg+1$-list coloring instances.
\end{lemma}
\begin{proof}

	Let $V_{\mathsf{rest}}\subseteq V(\hardcliques)$ be all vertices of hard cliques that are not in a slack triad and that have a neighbor that is not in a hard clique. Let $\tilde{G}$ be the subgraph induced by $V_{\mathsf{rest}}$. Observe that any hard clique of \typeIplus contains one uncolored slack vertex that is not part of $V_{\mathsf{rest}}$. Any clique of \typeI also contains an uncolored node not part of $V_{\mathsf{rest}}$. Hence, each node in $\tilde{G}$ has slack and we can color all nodes of $V_{\mathsf{rest}}$ with a single $deg+1$-list coloring instance.

	Next, we can color the remaining uncolored nodes in hard cliques, i.e., the slack vertices and vertices with an uncolored neighbor $\notin V(\hardcliques)$, with another $deg+1$-list coloring instance, as each of them has slack.

\end{proof}

\subsection{Runtime of Algorithm~\ref{alg:coloring}}\label{sec:runtimehard}

Let $T_{MM}(n, \Delta)$ denote the round complexity for computing a maximal matching,$T_{SP}$ for computing the degree splitting, and let $T_{deg+1}(n, \Delta)$ represent the round complexity to compute a $deg+1$-list coloring, both in a graph with $n$ vertices and maximum degree $\Delta$. Additionally, let $T_{HEG}(n)$ be the round complexity of computing a hyperedge grabbing solution in a hypergraph with $n$ vertices and minimum degree at least by a $1.05$ factor larger than the maximum rank.

\begin{lemma} \label{lem:runtime2}
	The round complexity of \Cref{alg:coloring} is upper bounded by
	\begin{align}
		O(T_{MM}(n,\Delta)+T_{SP}(n)+T_{deg+1}(n,\Delta)+T_{HEG}(n))~.
	\end{align} In particular the round complexity is upper bounded by $\min\{O(\Delta+\log n), \widetilde{O}(\log^{5/3}n)\}$.
\end{lemma}
\begin{proof}
	\Cref{alg:coloring} begins with Phase~1, which uses a maximal matching ($T_{MM}(n,\Delta)$) as a subroutine, followed by a HEG instance ($T_{HEG}(n)$). The sparsification is done in Phase~2 with a round complexity of $O\big(\eps^{-1}\cdot\log\eps^{-1}\cdot\big(\log\log\eps^{-1}\big)^{1.71}\cdot \log n\big)$ by \Cref{thm:degreeSplitting}. Setting up the slack triads in Phase~3 completes in a constant number of rounds, while Phase~4 requires a constant number of $deg+1$-list coloring instances ($T_{deg+1}(n,\Delta)$), as proven in \Cref{lemma:step4}. The time complexity of $T_{HEG}(n)$ is in $O(\log n)$ as we showed in \Cref{lem:rankMinDegree} that $\delta_H>1.1r_H$. For $\eps=1/63$ the round complexity $T_{SP}(n)$ is due \Cref{thm:degreeSplitting} in $O(\log n)$. The maximal matching and  (deg+1)-list coloring can either be solved in $\widetilde{O}(\log^{5/3}n)\}$ rounds with the algorithms from \cite{GG24}, or in $O(\Delta+\log^*n)$ rounds via the  algorithms from Linial, Panconesi and Rizzi, and Maus and Tonoyan \cite{linial92,panconesi-rizzi,MT20}.
\end{proof}

\subsection{Coloring Easy Cliques and Loopholes (Algorithm~\ref{alg:mainAlgorithm}, Line~\ref{step:coloreasycliques})}
\label{sec:easyLoophole}

\paragraph{Loophole Coloring} In this section, we describe \Cref{step:coloreasycliques}, which handles the coloring of vertices in easy almost-cliques and the remaining vertices in loopholes.
\begin{algorithm}[H]
	\caption{Coloring of loopholes and easy cliques}
	\label{alg:loophole}
	\begin{algorithmic}[1]
		\STATE Each loophole vertex votes for one of its loopholes;  $\mathcal{L}$  contains all loopholes with a vote
		\STATE Construct virtual graph $G_\mathcal{L}$ of $\mathcal{L}$
		\STATE Compute a $6$-ruling set $\cL_{rs}$ in $G_\mathcal{L}$\label[step]{step:loopholeRuling}
		\STATE Run BFS with depth $25$ from $\cL_{rs}$ to layer remaining vertices by depth.
		\FOR{$i = 25\ldots1$}
		\STATE Perform $(\deg+1)$-list coloring for vertices in layer $i$
		\ENDFOR
		\STATE Color the  vertices in $\cL_{rs}$ by bruteforce in $O(1)$ rounds
	\end{algorithmic}
\end{algorithm}

In the case where we parameterize our runtime by $n$ and $\Delta$ we use the following result (on a virtual graph) to compute the ruling set in \Cref{step:loopholeRuling}.
\begin{restatable}[Ruling Sets \cite{M21,SEW13}]{lemma}{thmRulingSets}\label{thm:rulingSet}
	For any constant integer $r\geq 2$, there is a deterministic \LOCAL algorithm that computes $(2,r)$-ruling in $O(\Delta^{\frac{2}{r+2}})+\log^* n$ rounds on any graph with maximum degree $\Delta$.
\end{restatable}
Let $T_{6-rs}(n,\Delta)$ be the round complexity to compute a $6$-ruling set.
\begin{lemma}\label{lem:easycoloring}
	\Cref{alg:loophole} completes a $\Delta$-coloring in $T_{4-rs}(n,\Delta)+O(T_{deg+1}(n,\Delta))$ rounds.
\end{lemma}
\begin{proof}
	We begin by limiting the total number of considered loopholes to $n$. To achieve this, each vertex votes for an arbitrary loophole it is in, and we define $\mathcal{L}$ as the set of loopholes that receive at least one vote.
	We define the following virtual graph instance
	$$G_\cL=(V_\cL, E_\cL)$$
	where $V_\cL$ is the set of loopholes $\cL$ and there exists an edge in $E_\cL$ between two vertices if the respective loopholes intersect or contain vertices that  are connected by an edge in $G$.

	While loopholes may overlap, we require a non-intersecting subset to avoid interference during coloring. To enforce this, we compute a $6$-ruling set on $G_\cL$. Since the algorithm only considers loopholes of size at most $3$, their diameter is bounded by $6$. Each vertex in a loophole can be adjacent to at most $\Delta$ other loopholes and contained in at most $\Delta^3$ loopholes, yielding an upper bound of $6(\Delta^3 + \Delta) \leq \Delta^4$ on the degree $\Delta'$ of $G_\cL$. We use \Cref{thm:rulingSet} to compute a $6$-ruling set on $G_\cL$.
	Once the non-intersecting loopholes—colored last—are selected, we first color the remaining vertices in easy almost-cliques. To do so, we perform a breadth first search (BFS) around each loophole up to a depth of $25$ in the original graph, assigning vertices to layers based on their lowest BFS depth.
	These $25$ layers contain all loophole vertices and also  all vertices in easy cliques as there are at most four loopholes of diameter $3$ between a vertex of such an easy clique and a loophole in $\cL$. Coloring proceeds from  layer $25$ downward to layer $1$, using $25$  $(deg+1)$-list coloring instances - one for each layer; note that each node has slack when coloring its layer due to at least one uncolored neighbor in the layer below.  We have that $T_{deg+1}(n,\Delta)=\min\{O(\Delta+\log n), \widetilde{O}(\log^{5/3}n)\}$.
	Finally, the loopholes themselves are colored via bruteforce, which is applicable due to their constant diameter of $3$.
\end{proof}

\subsection{Proof of Theorem~\ref{thm:DeltaColoring}}\label{sec:proofMainTheorem}
The validity of the coloring follows from \Cref{lemma:step4} and \Cref{lem:easycoloring}. \Cref{lemma:step4} shows that a valid $\Delta$-coloring can be computed for all vertices in hard cliques, while \Cref{lem:easycoloring} ensures that the $\Delta$-coloring can be completed for the remaining vertices in easy cliques and loopholes.

\textbf{Runtime.}
In Step 1, computing the ACD has a round complexity of $O(1)$ by \Cref{lem:acd}.  By \Cref{lem:runtime2}, the  step of \Cref{alg:mainAlgorithm} that colors vertices in hard cliques requires at most $\min\{O(\Delta+\log n), \widetilde{O}(\log^{5/3}n)\}$ rounds. The last two steps of coloring vertices in easy cliques and last the vertices in loopholes, detailed in \Cref{alg:loophole}, can be done in $O(\Delta)+\min\{O(\Delta+\log n), \widetilde{O}(\log^{5/3}n)\}$ rounds due to \Cref{lem:easycoloring}. Computing the $6$-ruling set requires $O(\Delta'^{\frac{2}{r+2}}) = O(\Delta^{\frac{8}{r+2}})= O(\Delta)$ rounds.
Combining these results, we conclude that the total round complexity is in $\min\{O(\Delta+\log n), \widetilde{O}(\log^{5/3}n)\}$.

\section{\texorpdfstring{Randomized $\Delta$-coloring}{Randomized Delta-coloring}}\label{sec:randomizedalg}
\begin{restatable}{theorem}{thmDeltaColoringRandomized}
	\label{thm:DeltaColoringRandomized}
	There is a randomized algorithm with runtime $\min\{\widetilde{O}(\log^{5/3}\log n), O(\Delta+\log\log n)\}$ that w.h.p.\ $\Delta$-colors  dense graphs.
\end{restatable}

If $\Delta=\omega(\log^{21}n)$: Use the $O(\log^* n)$-round  $\Delta$-coloring algorithm from \cite{FHM23}. Otherwise, we use the shattering framework that is at the core of essentially all recent randomized algorithms for distributed local graph problems. More precisely, we build on the shattering approach to $\Delta$-coloring from the current state-of-the-art randomized algorithm for constant-degree graphs by Ghaffari, Hirvonen, Kuhn, and Maus
\cite{GHKM21}; actually we use their algorithm that is entitled to be efficient for \emph{large} $\Delta$ in their work\footnote{Their algorithm for small $\Delta$ removes loopholes with radius up to $\log \log n$; the additional symmetry-breaking to color these loopholes requires $\omega(\log\log n)$ rounds. Their algorithm for large $\Delta$ is inefficient in other subroutines that we replace with subroutines of this paper. The benefit is that it only remove constant-sized loopholes. }. We next present the high-level overview of their algorithm. All steps, except for the post-shattering phase in \myStepCref{step:post-shattering} are taken black box from their work and explained in great detail in \cite{GHKM21}. We focus on the pre- and post-shattering phase in \myStepCref{step:pre-shattering}, and \ref{step:post-shattering}.

\begin{algorithm}
	\caption{Randomized $\Delta$-coloring of dense graphs }
	\begin{algorithmic}[1]
		\STATE If $\Delta=\omega(\log^{21}n)$: Use the $O(\log^* n)$-round  $\Delta$-coloring algorithm from \cite{FHM23}
		\STATE Pre-processing I: Identify small radius loopholes
		\STATE Pre-processing II: Compute an $O(1)$-ruling set on the virtual graph of loopholes\label{step:randomizedLoopholeRulingSet}
		\STATE Pre-processing III: Compute a $O(1)$-layering around loopholes in the ruling set and put loophoses and layers aside
		\label{step:randomizedLayers}
		\STATE Pre-Shattering: Randomly place non-adjacent $T$-nodes, same color the slack-pair with color one \label{step:pre-shattering}
		\STATE \textbf{**New** Post-shattering: } Use modified deterministic algorithm from \Cref{sec:deterministicDeltaColoring}/\Cref{thm:DeltaColoring} on small components in parallel \label{step:post-shattering}
		\STATE Post-processing I: Use $O(1)$ $deg+1$-list coloring instances to color layers from \Cref{step:randomizedLayers}
		\STATE Post-processing II: Color ruling set layers from \Cref{step:randomizedLoopholeRulingSet} in a brute-force manner in $O(1)$ rounds.
	\end{algorithmic}
\end{algorithm}

In the randomized pre-shattering phase they randomly place so-called $T$-nodes into the graph. These are identical to our slack triads, with the crucial difference that in their work these $T$-nodes  are non-adjacent. Hence, all slack pairs can be colored with the same color (for simplicity with the smallest color) providing slack to the respective slack vertex. Then, one forms a constant number layers of uncolored vertices  around the slack vertex with the same objective as in our work: coloring them in the reverse order. The slack vertex can be colored at the very end, as it has two same-colored neighbors, the first layer around the slack vertices forms a $deg+1$-list coloring as all vertices in the layer have slack due to the uncolored slack neighbor, and the second layer also has slack as the first layer is uncolored etc. They show that this process with a constant number of layers shatters the graph into components of size $N=\poly\Delta\cdot  \log n=O(\poly\log n)$ size. At this point \cite{GHKM21} uses a suboptimal algorithm to color these small components in time $O(\log^2 N)=O(\log^2\log n)$ (see Phase~6 in \cite{GHKM21}). This is the part that we replace with a slightly modified version of our deterministic algorithm from the previous section.

\paragraph{\myStepCref{step:post-shattering}: Modified post-shattering.} First, we extend the definition of a loophole to also comprise of nodes that have an uncolored neighbor outside of the small component; this change propagates to the definition of an easy almost cliques. Note that except for the slack pair in the $T$-nodes all nodes of the graph are still uncolored when handling the small components. Thus, any such loophole vertex also has slack. The second difference is due to adjacent slack pair vertices. For example, when a vertex $v$ at the boundary of the small component is adjacent to one vertex $u$ of a slack pair colored with the first color. Vertex $v$ does not have slack and hence cannot form a loophole. If $u$ is the only external neighbor of $v$ then $v$ cannot send a proposal to grab an edge for its clique; in fact vertex $v$ becomes useless for the clique. Similarly, if a vertex $u$ of a (hard) clique is a slack pair vertex colored with the first color, the clique cannot send a proposal through vertex $u$. In fact vertex $u$ becomes useless for the clique.  By chosing the distance between $T$-nodes to be large enough we can limit this to at most one useless vertex per clique. Note that the algorithm \cite{GHKM21} comes with an adjustable parameter $b$ to control this distance; $b$ can be chosen arbitrarily as long as it is constant.  Thus, in fact, each clique loses one potential proposal to grab an edge, but \Cref{eqn:HEG} has enough leeway to still solve the respective HEG instance in that case. Their pre-shattering phase has a pre-processing phase that removes constant-sized loopholes and their neighborhoods. Hence the $T$-nodes are only placed in hard cliques.

The rest of the proof is identical to the previous section, except that we use color space $\{2,\ldots,\Delta\}$ when coloring the slack pairs. Hence, there cannot be any conflicts with already existing colored slack pairs that only use the first color. Due to \Cref{lemma:triaddegree}, the maximum degree of the virtual instances of slack pairs is at most $\Delta-2$, and hence same-coloring the slack pairs with $\Delta-1$ colors is a $deg+1$-list coloring instance. When coloring the $O(1)$ layers around the slack vertices, we use all $\Delta$ colors as here nodes have slack regardless of whether they are adjacent to a node with the first color or not.

\paragraph{Runtime analysis $O(\Delta+\log\log n)$.} The runtime of the algorithm consists of several components. First, their pre-shattering phase has a pre-processing phase of removing constant-size loopholes that is also based on a ruling set computation, similar to how we deal with loopholes. By taking a sufficient $O(1)$-ruling set, the respective problem can be solved in $O(\Delta+\log^*n)$. The same runtime upper bound holds for the maximal matching computation of our deterministic algorithm in the small components, see \myStepCref{alg:coloring:1}.  Also the constant number of $deg+1$-list coloring and ruling set computations that our algorithm employs can be done in the same time \cite{MT20}; actually these can even be done in $O(\sqrt{\Delta\log\Delta}+\log^* n)$ rounds.
Additionally, we have $O(\log N)=O(\log\log n)$ rounds for solving the HEG subroutine and for sparsifying the computed matching, see \myStepCref{alg:coloring:2} and Phase~2 of \Cref{alg:coloring}.
Thus, the total runtime is upper bounded by $O(\Delta+\log\log n)$.

\paragraph{Runtime analysis $\widetilde{O}(\log^{5/3}\log n)$.}
Alternatively, we can also upper bound the runtime as $\widetilde{O}(\log^{5/3}\log n)$. Whenever $\Delta=\omega(\log^{21}n)$ we use the randomized $\Delta$-coloring algorithm from \cite{FHM23} that can solve the problem in $O(\log^{*}n)$ rounds. In the other case the small components in the post-shattering phase are of size $N=O(\poly\Delta \cdot \log n)=O(\poly\log n)$. We use the deterministic algorithms from \cite{GG24} for maximal matching and for the constant number of $deg+1$-list coloring instances. The invocation of the ruling set algorithm is replaced by the MIS algorithm, i.e., a $1$-ruling set, from the same paper. Each of these algorithms are deterministic and run in $\widetilde{O}(\log^{5/3}N)=\widetilde{O}(\log^{5/3}\log n)$ rounds on the small components. Lastly, observe that the pre-processing of the pre-shattering phase also involved the computation of a certain ruling set on a virtual graph with maximum degree $\poly \Delta$. This is also replaced with the MIS algorithm from \cite{GG24} combined with Ghaffari's seminal MIS algorithm \cite{GhaffariImproved16} resulting in a runtime of $O(\log (\poly\Delta)) +\widetilde{O}(\log^{5/3}\log n)=\widetilde{O}(\log^{5/3}\log n)$ rounds. Note, that in this latter argument we used that $\Delta\leq \poly\log n$, as otherwise the runtime of the preshattering phase is $O(\log \poly\Delta)\gg \poly\log \log n$, and also the post-shattering phase of Ghaffari's MIS algorithm cannot profit from the polylogarithmic component size (as components are not actually polylogarithmic for large $\Delta$).

\bibliographystyle{alpha}
\bibliography{references}

\appendix

\section{Further related work}
\label{sec:relatedWork}
The first distributed algorithms for $\Delta$-coloring were introduced by Panconesi and Srinivasan in the 90ies \cite{PS95}. Ghaffari, Hirvonen, Kuhn, Maus improved upon their results by presenting several algorithms for the problem with different tradeoffs, including a randomized algorithm with runtime $O(\log\Delta+\poly\log\log n)$ \cite{GHKM21}. The runtime can be achieved by using faster algorithms for the $deg+1$-list coloring subroutines, e.g., the one from \cite{HKNT22}. Their work also contains the state-of-the-art algorithms for $\Delta$-coloring constant-degree graphs, namely an $O(\log^2 n)$-round deterministic and an $O(\log^2\log n)$-round randomized algorithm.
Recently, Fischer, Halldorsson, and Maus presented the first randomized sublogarithmic-time algorithm for the problem in the \LOCAL model for general graphs, namely a $\poly\log\log n$-round algorithm \cite{FHM23}. Their work also $\Delta$-colors in as few as $O(\log^* n)$ rounds when $\Delta=\omega(\log^{21}n)$.  Maus and Uitto obtain a randomized $\poly\log\log n$-round \CONGEST algorithm for the problem by deducing a more general result on efficient algorithms for the \lovasz Local Lemma \cite{MU21}. Very recently,  Halldorsson and Maus developed a randomized  $\poly\log \log n$-round  \CONGEST algorithm for general graphs \cite{HM24}. In \cite{CCDM24} Coy, Czumaj, Davies, and Mishra implemented the $\Delta$-coloring algorithm from \cite{FHM23} to obtain an algorithm in the massively parallel computing model. Assadi, Kumar, and Mittal designed a single-pass semi-streaming algorithm for $\Delta$-coloring \cite{AKM22}. In their search for a coloring algorithm that uses fewer than $\Delta$ colors, Bamas and Esperet \cite{BamasEsperet19} introduced a polylogarithmic randomized algorithm achieving a $\Delta - \sqrt{\Delta}$ coloring.

\section{Missing Proofs from Section~\ref{sec:deterministicDeltaColoring}}
\label{sec:missingProofs}
\lemmasetsize*
\begin{proof}
	We start by proving that $|f(V(Q_{C}^i))| = |Q_{C}^i|$. Assume $f(u) = f(v)$ for $u,v\in Q_{C}^i$ then $F(v)$ is adjacent to $u$ and $v$ which is a contradiction to \Cref{lem:hardcliqueprops:notwons}, Part~\ref{lem:hardcliqueprops}.

	To prove the second equality $|\phi(V(Q_{C}^i))| = |Q_{C}^i|$, we use a similar argument as in the first part. Assume for contradiction, that $\phi(u) = \phi(v)$ for two distinct vertices $u, v \in Q_{C}^i$, where both vertices request to grab the same edge $\phi(u) = \phi(v) = \{e_1, e_2\}$. If $f(u) = f(v)$, then this case directly reduces to the first part of the proof, and the claim follows. Otherwise, $u$ and $v$ must be adjacent to different endpoints of the edge $\phi(u) = \{e_1, e_2\}$. Let $c_1$ and $c_2$ be two vertices in $C$, as defined in the first part: $c_1 \in C \setminus \{u, v\}$ is a vertex not adjacent to $f(u)$, and $c_2 \in C$ is a vertex adjacent to both $u$ and $c_1$. With this, the vertices $u, e_1, e_2, v, c_1, c_2$ again form a loophole that intersects the clique $C$. This contradicts the assumption that $C$ is a hard clique, which by definition cannot intersect any loophole. Thus, the assumption that $\phi(u) = \phi(v)$ leads to a contradiction, and we conclude that $|\phi(V(Q_{C}^i))| = |Q_{C}^i|$.
\end{proof}

\lemmahardcliqueprops*
\begin{proof}
	\textbf{Property~1:} Each vertex of a hard clique $C$ has degree exactly $\Delta$, otherwise the clique would be classified as easy clique based on \Cref{def:loopholes}. If two nodes $u_1,u_2\in C$ were not connected, they together with two adjacent vertices $u_3,u_4$ in $N(u_1)\cap N(u_2)$ would form a loophole. See \Cref{fig:loophole}. The vertex $u_3$ exists for the following reason: As $u_1$ and $u_2$ each have $(1-\eps)\Delta$ neighbors in $C$, there are at most $\eps\Delta$ vertices $A_1\subseteq C$ not connected to $u_1$ and at most $\eps\Delta$ vertices $A_2\subseteq C$ not connected to $u_2$. Hence, $|C\setminus (A_1\cup A_2)|\geq |C|-2\eps\Delta \geq (1-2.25\eps)\Delta>0$, where we also used that $|C|\geq (1-\eps/4)\Delta$ by the properties of ACs given in \Cref{lem:acd}. Let $u_3\in C\setminus (A_1\cup A_2)$ be any vertex in that set. Similarly observe that there are at most $\eps\Delta$ vertices $A_3\subseteq C$ not connected to $A_3$. Then $u_4$ can be an arbitrary vertex in $C\setminus (A_1\cup A_2\cup A_3)$ that exists as $|C\setminus (A_1\cup A_2\cup A_3)|\geq (1-3.25\eps)\Delta$.
	\textbf{Property~2:} As $C$ forms a clique and each vertex $v\in C$ has degree $\Delta$, each vertex has has $e_C=\Delta - |C| + 1$ neighbors outside its clique.
	\textbf{Property~3:}
	Assume for contradiction  that some $w\notin C$ is adjacent to two distinct vertices $u$ and $v$ of $C$.  Let $c_1 \in C \setminus \{u, v\}$ be a vertex inside $C$ that is not adjacent to $w$. Note that such a vertex $c_1$ must exist; otherwise, $w$ would be a part of $C$ due to \Cref{lem:acd}, property iii). Further, let $c_2 \in C$ be a vertex adjacent to both $u$ and $c_1$. Given this, the vertices $u, v, c_1, c_2$ form a loophole (\Cref{def:loopholes}), a contradiction to $C$ being hard.
\end{proof}
\begin{figure}
	\centering
	\scalebox{0.7}{\begin{tikzpicture}
	\def\xO{0}
	\def\yO{0}
	\def\R{3} 
	\def\Ri{2.4}
	\def\rA{0.35}
	\def\N{6} 
	
		\draw[thick, gray] (\xO,\yO) circle(\R cm);
		
		\foreach \i in {0,1,...,\N} {
			\pgfmathsetmacro{\theta}{360/\N*\i} 
			\pgfmathsetmacro{\x}{\Ri*cos(\theta)}
			\pgfmathsetmacro{\y}{\Ri*sin(\theta)}
			\coordinate (P\i) at (\x,\y);
		}

	\foreach \i in {0,1,...,\N} {
		\foreach \j in {0,1,...,\N} {
			\ifnum \i<\j 
				\draw[color=gray] (P\i) -- (P\j);
				\draw[color=gray] (P\j) -- (P\i);
                
			\fi
		}
	}

    \foreach \i in {0,2,3,4} {
        \draw[] (P\i) -- ++(360/\N*\i-15:1.5cm);
        \draw[] (P\i) -- ++(360/\N*\i+15:1.5cm);
    }
    \draw[] (P1) -- ++(360/\N*1+15:1.5cm);
    \draw[] (P5) -- ++(360/\N*5+15:1.5cm);
    \begin{scope}
        \clip (6,3) rectangle (8,0);
        \draw[thick, gray] (9.4,2) circle(\R cm);
    \end{scope}
    \draw[dashed, color=red, line width = 2pt] (7,2) -- (P3);
    \draw[color=red, line width = 2pt] (7,2) -- (P1);
    \draw[color=red, line width = 2pt] (7,2) -- (P5);   
    \draw[color=red, line width = 2pt] (P1) -- (P3); 
    \draw[color=red, line width = 2pt] (P3) -- (P5); 
    \foreach \i in {0,1,...,\N} {
		\draw[fill=white] (P\i) circle(\rA cm);
	}

    \draw[color=black, fill = white, line width= 2pt] (7,2) circle(\rA cm);
    
\end{tikzpicture}}
	\caption{Example of a loophole (red circle) that is used in the proof of \Cref{lem:hardcliqueprops}, Part~\ref{lem:hardcliqueprops:notwons}.}
	\label{fig:loophole}
\end{figure}
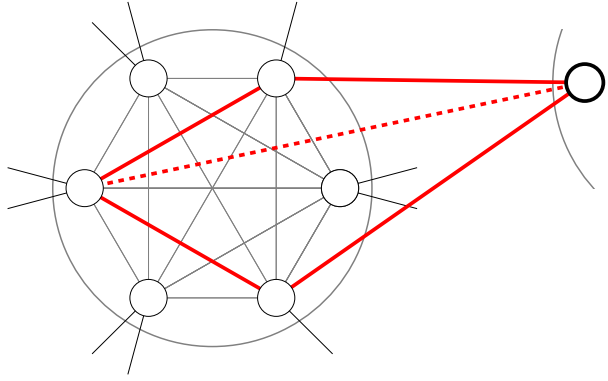
\section{Complementary Subroutines}
\label{sec:subroutines}
\subsection{Degree Splitting}
\begin{lemma}\label{thm:degreeSplitting}
	For every $\eps>0$, there are deterministic $O\big(\eps^{-1}\cdot\log\eps^{-1}\cdot\big(\log\log\eps^{-1}\big)^{1.71}\cdot \log n\big)$-round distributed algorithms for computing directed and undirected degree splittings with the following properties:
	\begin{enumerate}
		\item For directed degree splitting, the discrepancy at each node $v$ of degree $d(v)$ is at most $\eps \cdot d(v) + 1$ if $d(v)$ is odd and at most $\eps\cdot  d(v) + 2$ if $d(v)$ is even.
		\item For undirected degree splitting, the discrepancy at each node $v$ of degree $d(v)$ is at most $\eps\cdot d(v) + 4$.
	\end{enumerate}
\end{lemma}
\begin{corollary} \label{cor:degsplit}
	For $i\in\mathbb{N}_{>0}$ and given $\eps>0$, there exists a deterministic $O(\log n)$ round algorithm to split the edges of a graph into $k=2^i$ parts such that the following holds for any of the $k$ parts:

	For each node $v$ the number of edges contained in the part and incident to $v$ lies in the interval
	\begin{align}[deg(v)/2^i-\eps d(v)-a,deg(v)/2^i+\eps d(v)+a]~, \end{align}where $a=2\sum_{j=0}^{i-1}(1/2+\eps/4)^j$. \end{corollary}
\begin{proof}
	Pick an $\eps'=\eps/2$ and do the analysis like on the board of bounded the max and min degree by induction and use \Cref{claim:splittingDiscrapancy}
	We prove the claim by induction on $i$ on the iterations that we get the following after $i$ recursive splits on the previous result.
	\begin{align*}a=2\sum_{j=0}^{i-1}(1/2+\eps/4)^j=2\cdot \frac{(\frac{1}{2}+\eps/4)^i}{\frac{1}{2}-\eps/4}=4\cdot(\frac{1}{2}+\eps/4)^i(1-\eps/2)^{-1} & \qedhere \end{align*}
\end{proof}

\begin{claim}
	\label{claim:splittingDiscrapancy}
	$(1/2+\eps)^n\leq 1/2^n+2\epsilon$ holds for integers $n\geq 1$ and $0\leq \eps\leq 1/2$.
\end{claim}
\begin{proof}
	\textbf{Induction start ($n=1$):} $(1/2+x)\leq 1/2+2x$ holds.

	\textbf{Induction step ($n\rightarrow n+1$)}
	\begin{align*}
		\left(\frac{1}{2}+\eps\right)^{n+1} & =\left(\frac{1}{2}+\eps\right)\cdot \left(\frac{1}{2}+\eps\right)^n\stackrel{I.H.}{\leq} \left(\frac{1}{2}+\eps\right) \left(\frac{1}{2^n}+2x\right) \\
		                                    & = 1/2^{n+1}+x+x\cdot (1/2^n+2x)\leq 1/2^{n+1}+2x,
	\end{align*}
	where we used $1/2^n+2x\leq 1$ for $x\leq 1/2$ and $n\geq 1$ in the last step.
\end{proof}

\subsection{List Coloring}

\begin{lemma}[List coloring \cite{MT20}]
	\label{lem:listColoring}
	There is a deterministic distributed algorithm to $(deg+1)$-list-color any graph with maximum degree $\Delta$ in $O(\sqrt{\Delta\log\Delta})+O(\log^* n)$ rounds.
\end{lemma}

\end{document}